\documentclass[11pt]{article}
\usepackage{amsmath,amssymb,bbm,amsthm}
\usepackage{fullpage}
\usepackage{thm-restate,color,xcolor,xspace}
\usepackage{hyperref,cleveref}
\usepackage{graphicx}
\usepackage{algorithm,algorithmic}
\usepackage{authblk}

\newtheorem{theorem}{Theorem}
\newtheorem{lemma}[theorem]{Lemma}
\newtheorem{definition}[theorem]{Definition}
\newtheorem{corollary}[theorem]{Corollary}

\newtheorem{claim}[theorem]{Claim}

\newtheorem{fact}[theorem]{Fact}

\begin{document}

\title{Shortcutting for Negative-Weight Shortest Paths}


\author[1]{George Z. Li}
\author[1]{Jason Li}
\author[2]{Satish Rao}
\author[3]{Junkai Zhang}
\affil[1]{Carnegie Mellon University $\{\texttt{gzli}, \texttt{jmli}\}$\texttt{@cs.cmu.edu}}
\affil[2]{UC Berkeley \texttt{satishr@berkeley.edu}}
\affil[3]{Tsinghua University \texttt{zhangjk22@mails.tsinghua.edu.cn}}

\date{}
\maketitle

\begin{abstract}
Consider the single-source shortest paths problem on a directed graph with real-valued edge weights. We solve this problem in $O(n^{2.5}\log^{4.5}n)$ time, improving on prior work of Fineman (STOC 2024) and Huang-Jin-Quanrud (SODA 2025, 2026) on dense graphs. Our main technique is an shortcutting procedure that iteratively reduces the number of negative-weight edges along shortest paths by a constant factor.
\end{abstract}

\section{Introduction}

We consider the problem of computing single-source shortest paths on a directed, edge-weighted graph. While this problem is among the oldest in theoretical computer science, the past few years have witnessed exciting breakthroughs in multiple different settings. When all edge weights are non-negative, the classic Dijkstra's algorithm solves the problem in $O(m+n\log n)$ time, which has been improved to $O(m\sqrt{\log n\log\log n})$ on undirected graphs~\cite{duan2023randomized} and $O(m\log^{2/3}n)$ on directed graphs~\cite{duan2025breaking}. When all edge weights are integral (and possibly negative), the problem was solved in $\tilde O(m)$ time\footnote{$\tilde O(\cdot)$ ignores factors poly-logarithmic in $n$.}~\cite{bernstein2025negative} and $m^{1+o(1)}$ time~\cite{chen2025maximum} by two very different approaches, with the latter solving the more general problem of minimum-cost flow. Finally, in the general case of real-valued edge weights, the $O(mn)$ time Bellman-Ford algorithm remained the fastest for 70~years until the recent breakthrough of Fineman~\cite{fineman2024single}, who obtained an $\tilde O(mn^{8/9})$ time algorithm. The running time has been improved further to $\tilde O(mn^{4/5})$~\cite{huang2025faster} and then $\tilde O(mn^{3/4}+m^{4/5}n)$~\cite{huang2026faster}. All three of these results rely on new techniques introduced by Fineman~\cite{fineman2024single}, including betweenness reduction, negative sandwich, and remote sets.

In this work, we introduce the concept of \emph{shortcutting} to the negative-weight shortest paths problem. Our method adds \emph{shortcut} edges to the graph so that the number of negative-weight edges along shortest paths is reduced by a constant factor. Together with a \emph{recursive} implementation of Fineman's betweenness reduction, we obtain an $\tilde O(n^{2.5})$ time algorithm for this problem. Notably, our algorithm does not make use of negative sandwiches or remote sets, resulting in a more ``natural'' $\sqrt n$-factor from optimizing the betweenness reduction tradeoff.

\begin{theorem}\label{thm:main}
There is a randomized algorithm for single-source shortest paths on real-weighted graphs that runs in $O(n^{2.5}\log^{4.5}n)$ time.
\end{theorem}

Additionally, we show that our idea of recursive betweenness reduction can also speed up existing algorithms, even on sparse graphs. We directly apply our recursive betweenness reduction to the algorithm of~\cite{huang2025faster}, obtaining the following slight improvement for sparse graphs. 

\begin{theorem}\label{thm:main-sparse}
There is a randomized algorithm for single-source shortest paths on real-weighted graphs that runs in $\tilde{O}(mn^{7/9})$ time.
\end{theorem}

Using additional ideas to refine the approach of~\cite{huang2026faster}, we can get an even faster algorithm. This should be compared to that of~\cite{huang2026faster}, which runs in $\tilde{O}(mn^{3/4}+m^{4/5}n)$. Since this requires more work on top of the recursive betweenness reduction, we defer the proof to the Appendix.

\begin{theorem}
    There is a randomized algorithm for single-source shortest path on real-weighted graphs that runs in time $O(mn^{0.695}+m^{0.766}n)$.
\end{theorem}
While the improvement is minor, our aim is to demonstrate the versatility of our techniques when applied to existing algorithms. We emphasize that \Cref{thm:main} is still our main result, and the focus is on conceptual ideas over the optimization of constants.

\subsection{Our Techniques}

Our starting point is a new interpretation of Fineman's betweenness reduction~\cite{fineman2024single}. The goal of betweenness reduction is to reweight the edges of the graph so that for all pairs of vertices $s,t\in V$, there are relatively few vertices $v\in V$ satisfying $\tilde d(s,v)+\tilde d(v,t)<0$, where $\tilde d$ is a certain distance metric that we do not define for ease of exposition. We say that such a vertex $v$ is \emph{between} the pair $(s,t)$. In particular, the total number of triples $(s,t,v)$ for which $v$ is between $(s,t)$ is much less than $n^3$. Now for each vertex $r\in V$, consider the sets $V_r^{out}=\{v\in V\mid\tilde d(r,v)\le0\}$ and $V_r^{in}=\{v\in V\mid\tilde d(v,r)<0\}$. Observe that for each pair $s\in V_r^{in}$ and $t\in V_r^{out}$, the vertex $r$ is between the pair $(s,t)$. It follows that $r$ is between at least $|V_r^{in}|\cdot|V_r^{out}|$ pairs $(s,t)$. Summing over all $r$, the number of triples $(s,t,r)$ for which $r$ is between $(s,t)$ is at least $\sum_r|V_r^{in}|\cdot|V_r^{out}|$ and much less than $n^3$. So for an average vertex $r\in V$, the product $|V_r^{in}|\cdot|V_r^{out}|$ is much less than $n^2$.

In the ideal case, we have $|V_r^{in}|\approx|V_r^{out}|$, in which case both $V_r^{in}$ and $V_r^{out}$ are relatively small in size for an average vertex $r$. We then compute the sets $V_r^{in}$ and $V_r^{out}$ in sublinear time, and add \emph{shortcut} edges that attempt to shortcut any path that contains the vertex $r$. The shortcutting procedure is technical, so we skip it in the overview.

In general, it is possible that $V_r^{in}$ and $V_r^{out}$ have very different sizes. For example, even if $|V_r^{in}|\cdot|V_r^{out}|\ll n^2$, it is possible that $|V_r^{in}|$ is small but $|V_r^{out}|=n$, in which case it is not possible to compute $V_r^{out}$ in sublinear time. Our key idea is to consider a threshold $\Delta_r$ instead, and define $V_r^{out}=\{v\in V\mid\tilde d(r,v)\le-\Delta_r\}$ and $V_r^{in}=\{v\in V\mid\tilde d(v,r)<\Delta_r\}$. It is still true that vertex $r$ is between any pair $(s,t)$ with $s\in V_r^{in}$ and $t\in V_r^{out}$, so all previous bounds apply. Observe that $|V_r^{out}|$ and $|V_r^{in}|$ are monotone decreasing and increasing, respectively, as $\Delta_r$ increases, so there is a sweet spot where $|V_r^{out}|\approx|V_r^{in}|$, at least when ignoring the possibility of many vertices taking the same value $\tilde d(v,r)$ or $\tilde d(r,v)$. Moreover, the distance metric $\tilde d$ is simple enough that a forward Dijkstra and a separate backward Dijkstra can find the ideal value $\Delta_r$ and the sets $V_r^{out}$ and $V_r^{in}$. Overall, we compute the sets $V_r^{out}$ and $V_r^{in}$ in sublinear time for an average vertex $r$, which is $\ll mn$ time over all vertices $r$.

For technical reasons, the distance metric $\tilde d$ is defined so that it suffices to only look at \emph{negative} vertices $r$, which are vertices with at least one negative outgoing edge. Therefore, our running time becomes much faster when there are few negative vertices. Also, we need to repeat the entire shortcutting procedure $O(\log n)$ times to fully shortcut all shortest paths. The final graph may be dense even if the initial graph is sparse, which explains why our algorithm performs better on dense graphs.

\subsection{Preliminaries}
All graphs in this paper are directed and weighted, possibly with negative edge weights. Given a graph $G=(V,E)$ and a vertex $u\in V$, define $N^{out}(u)$ and $N^{in}(u)$ as the out-neighbors and in-neighbors of $u$, i.e., $N^{out}(u)=\{v\in V:(u,v)\in E\}$ and $N^{in}(u)=\{v\in V:(v,u)\in E\}$.

Following~\cite{fineman2024single}, we characterize $(s,t)$-paths by their number of negative (weight) edges. For a positive integer $h$ called the \emph{hop bound}, define the \emph{$h$-negative hop distance} $d^h(s,t)$ as the minimum weight $(s,t)$-path with at most $h$ negative edges. Given a source vertex $s$, we can compute $d^h(s,t)$ for all other $t\in V$ in $O(h(m+n\log n))$ by a hybrid of Dijkstra and Bellman-Ford~\cite{dinitz2017hybrid}.

Similar to prior work~\cite{fineman2024single,huang2025faster,huang2026faster}, our algorithm makes use of \emph{potential functions}, which are functions $\phi:V\to\mathbb R$ that re-weight the graph, where each edge $(u,v)$ has new weight $w_\phi(u,v)=w(u,v)+\phi(u)-\phi(v)$. The key property of potential functions is that for any vertices $s,t\in V$, all $(s,t)$-paths have their weight affected by the same additive $\phi(s)-\phi(t)$. A potential $\phi$ is \emph{valid} if $w_\phi(u,v)\ge0$ for all edges $(u,v)$ with $w(u,v)\ge0$, i.e., it does not introduce any new negative edges.

Define a \emph{negative vertex} as a vertex with at least one negative outgoing edge. We will use the following routine implicit in Fineman~\cite{fineman2024single}.
\begin{lemma}\label{lem:one-negative-outgoing-edge}
Given a graph $G=(V,E)$ with $k$ negative vertices, there is a linear-time algorithm that outputs a graph $G'=(V',E')$ with $V'\supseteq V$ such that
 \begin{enumerate}
 \item For any $s,t\in V$ and hop bound $h$, we have $d_{G'}^h(s,t)=d_G^h(s,t)$,
 \item Each negative vertex has exactly one negative outgoing edge,
 \item $G'$ also has $k$ negative vertices, and
 \item $|V'|=|V|+k$.
 \end{enumerate}
\end{lemma}
\begin{proof}
For each negative vertex $u$, consider its outgoing edges $(u,v_1),\ldots,(u,v_\ell)$, and suppose that $(u,v_1)$ is the edge of smallest weight, with ties broken arbitrarily. Create a new vertex $u'$, add the edge $(u,u')$ of weight $(u,v_1)$, and replace each edge $(u,v_i)$ by the edge $(u',v_i)$ of weight $w_G(u,v_i)-w_G(u,v_1)\ge0$. It is straightforward to verify all of the required conditions.
\end{proof}

Similar to Fineman~\cite{fineman2024single}, we may \emph{freeze} the set of negative edges before re-weighting by a potential $\phi$, so that any negative edges that are accidentally re-weighted to be non-negative are still considered negative edges. In particular, for any hop bound $h$, define $d^h_\phi(s,t)$ as the $h$-negative hop distance under the re-weighting, which satisfies $d^h_\phi(s,t)=d^h(s,t)+\phi(s)-\phi(t)$ as long as negative edges are frozen.

\section{Recursive Betweenness Reduction}

One of Fineman's new insights is the concept of \emph{betweenness reduction}~\cite{fineman2024single}. In this paper, we use a specialized definition of betweenness that is easier to work with. First, define $d^-(s,v)$ as the shortest distance among $s\to v$ paths with all edges except the first being non-negative, i.e., only the first edge is allowed to be negative. In particular, $d^0(s,v)\ge d^-(s,v)\ge d^1(s,v)$. We say that a vertex $v$ is \emph{between} the ordered pair $(s,t)$ if $d^0(s,v)+d^-(v,t)<0$. Note that any vertex $v$ between $(s,t)$ must be negative, since $d^0(s,v)+d^-(v,t)<0$ implies that $d^-(v,t)<0$, which means $v$ must have a negative outgoing edge. The \emph{betweenness} of an ordered pair $(s,t)$ is the number of (negative) vertices $v\in V$ between $(s,t)$.

The benefit of our specific definition of betweenness is that \emph{betweenness reduction} is much easier to compute, compared to the general procedure from~\cite{fineman2024single}.

\begin{lemma}[Betweenness reduction]\label{lem:betweenness-reduction}
Consider a graph $G$ with $k$ negative vertices, and freeze the negative edges. For any parameter $b\ge1$, there is a randomized algorithm that returns either a set of valid potentials $\phi$ or a negative cycle, such that with high probability, all pairs $(s,t)\in V\times V$ have betweenness at most $k/b$ under the new weights $w_\phi$. The algorithm makes one call to negative-weight single-source shortest paths on a graph with $O(b\log n)$ negative vertices.
\end{lemma}
\begin{proof}
The algorithm first samples $O(b\log n)$ negative vertices into a set $S$. Next, construct a graph $H$ starting from the non-negative edges of $G$, and then adding in the negative out-going edges from vertices in $S$. By construction, $H$ has $O(b\log n)$ many negative vertices. The algorithm makes one call to negative-weight single-source shortest paths in $H$. If a negative cycle is returned, then it is also a negative cycle in $G$, and we are done. Otherwise, a (valid) potential $\phi$ is returned such that $w_\phi(u,v)\ge0$ for all edges $(u,v)$ in $H$.

We claim that with high probability, all pairs $(s,t)\in V\times V$ have betweenness at most $k/b$ under the new weights $w_\phi$. Consider a fixed pair $(s,t)$, order the vertices $v\in V$ by value of $d^0(s,v)+d^-(v,t)$ and consider the $k/b$ vertices with the smallest such values. With high probability, there is a sampled vertex $x$ among these $k/b$ vertices. In this case, since negative edges are frozen, we have $d^0_\phi(s,v)+d^-_\phi(v,t)=d^0(s,v)+d^-(v,t)+\phi(s)-\phi(t)$. In particular, the ordering does not change if we order by $d^0_\phi(s,v)+d^-_\phi(v,t)$ instead, so the sampled vertex $x$ is still among the $k/b$ lowest when ordered by $d^0_\phi(s,v)+d^-_\phi(v,t)$.

Since $x\in S$, all outgoing edges from $x$ are in $H$ and are therefore non-negative under $\phi$, so $d^-_\phi(x,t)\ge0$. It follows that $d^0_\phi(s,x)+d^-_\phi(x,t)\ge0$. Any vertex $u\in V$ that is between $(s,t)$ must have its value lower than $d^0_\phi(s,x)+d^-_\phi(x,t)$, of which there are at most $k/b$.
\end{proof}
Naively, negative-weight shortest paths on a graph with $O(b\log n)$ negative vertices can be computed in $\tilde O(bm)$ time, since it suffices to compute $O(b\log n)$-negative hop shortest paths. While this naive bound matches the running time bottleneck from~\cite{fineman2024single}, we can speed it up by \emph{recursively} calling negative-weight shortest paths. We analyze this recursive procedure in \Cref{sec:recurrence}.

\section{Iterative Hop Reduction}

In this section, we describe our hop reduction procedure. We first apply \Cref{lem:one-negative-outgoing-edge} to the input graph so that each negative vertex has exactly one negative outgoing edge, creating $k$ new vertices. We then freeze the negative edges and apply the betweenness reduction from \Cref{lem:betweenness-reduction} with a certain parameter $b$. We then run the following hop reduction theorem on the re-weighted graph, which is the main focus of this section.
\begin{theorem}\label{thm:single-step}
Consider a weighted, directed graph $G=(V,E)$ with $k$ negative vertices, each with exactly one (frozen) negative outgoing edge. Given a parameter $b\ge1$, suppose that the guarantee of \Cref{lem:betweenness-reduction} holds, i.e., all pairs $(s,t)\in V\times V$ have betweenness at most $k/b$. There is a deterministic $O(kn^2/\sqrt b)$-time algorithm that outputs a supergraph $G'=(V',E')$ with $V'\supseteq V$ and $E'\supseteq E$, with potentially new negative edges, such that
 \begin{enumerate}
 \item For any $s,t\in V$, we have $d_{G'}(s,t)=d_G(s,t)$,\label{item:single-step-1}
 \item For any $s,t\in V$ and hop bound $h$, we have $d^{h-\lfloor h/3\rfloor}_{G'}(s,t)\le d^h_G(s,t)$,\label{item:single-step-2}
 \item $|V'|=n+k$, and there are no new negative vertices.\label{item:single-step-3}
 \end{enumerate}
\end{theorem}

Before we prove \Cref{thm:single-step}, we first show how to apply it iteratively to solve negative-weight shortest paths. By property~(\ref{item:single-step-3}), the new graph $G'$ has $k$ additional vertices, and the number of negative vertices is still $k$. By property~(\ref{item:single-step-2}) with $h=k$, we have $d^{k-\lfloor k/3\rfloor}_{G'}(s,t)\le d^k_G(s,t)=d_G(s,t)$ for all $s,t\in V$. We then un-freeze the negative edges so that $d^{k-\lfloor k/3\rfloor}_{G'}(s,t)$ can only be smaller, and repeat the hop-reduction step from the beginning, starting with \Cref{lem:one-negative-outgoing-edge}. Each loop reduces the number of negative hops along shortest paths by a constant factor, so after $O(\log k)$ iterations, all shortest paths have at most $2$ negative edges. Finally, running $2$-negative hop shortest paths correctly computes single-source distances.

We now bound the running time. Each iteration introduces at most $k$ negative edges from \Cref{lem:one-negative-outgoing-edge} and at most $k$ from \Cref{thm:single-step}. (The number is \emph{at most} $k$ because the number of negative vertices may decrease after un-freezing negative edges.) Therefore, the final graph has $n+O(k\log k)$ vertices. The final running time is dominated by $O(\log k)$ calls to betweenness reduction on graphs of $n+O(k\log k)$ vertices, plus $O(\log k\cdot k(n+k\log k)^2/\sqrt b)$ additional time. Together with the running time guarantee of \Cref{lem:betweenness-reduction}, we obtain the following:

\begin{corollary}\label{cor:recurrence}
Consider a weighted, directed graph $G=(V,E)$ with $k$ negative vertices. For any parameter $b\ge1$, there is a negative-weight single-source shortest paths algorithm that makes $O(\log k)$ recursive calls on graphs with $n+O(k\log k)$ vertices and $O(b\log n)$ negative vertices, and runs in $O(k\log k\cdot(n+k\log k)^2/\sqrt b)$ additional time.
\end{corollary}

In \Cref{sec:recurrence}, we solve the recurrence above to $O(\sqrt k(n+k\log k)^2\log^{2.5}n)$. For the rest of this section, we instead focus on proving \Cref{thm:single-step}.

Our main insight is the following novel routine, which is accomplished by running Dijkstra's algorithm in both directions from a source vertex $r$.

\begin{lemma}\label{lem:dijkstra-both-ways}
There is an algorithm that, for any given vertex $r\in V$, computes a number $\Delta_r$ and two sets $V_r^{out}$ and $V_r^{in}$ such that
 \begin{enumerate}
 \item $d^-(r,v)\le-\Delta_r$ for all $v\in V_r^{out}$,\label{item:dijkstra-both-ways-1}
 \item $d^-(r,v)\ge-\Delta_r$ for all $v\notin V_r^{out}$,\label{item:dijkstra-both-ways-2}
 \item $d^0(v,r)\le\Delta_r$ for all $v\in V_r^{in}$,\label{item:dijkstra-both-ways-3}
 \item $d^0(v,r)\ge\Delta_r$ for all $v\notin V_r^{in}$, and\label{item:dijkstra-both-ways-4}
 \item Either (\ref{item:dijkstra-both-ways-1}) or (\ref{item:dijkstra-both-ways-3}) is satisfied with strict inequality. That is, either $d^-(r,v)<-\Delta_r$ for all $v\in V_r^{out}$, or $d^0(v,r)<\Delta_r$ for all $v\in V_r^{in}$.\label{item:dijkstra-both-ways-5}
 \item $\big||V_r^{out}|-|V_r^{in}|\big|\le1$.\label{item:dijkstra-both-ways-6}
 \end{enumerate}
The algorithm runs in time $O((|V_r^{out}|+|V_r^{in}|)^2+(|V_r^{out}|+|V_r^{in}|)\log n)$. Moreover, the algorithm can output the values of $d^-(r,v)$ for all $v\in V_r^{out}$ and $d^0(v,r)$ for all $v\in V_r^{in}$.
\end{lemma}
\begin{proof}
We run two Dijkstra's algorithms in parallel on two separate graphs, both with source $r$. The first graph $H_1$ consists of the non-negative edges of $G$ (at the time of freezing) together with the (possibly negative) edges $(r,v)$. (Recall that Dijkstra's algorithm still works from source $r$ if all negative edges are of the form $(r,v)$.) The second graph $H_2$ consists of the non-negative edges of $G$, but reversed. Recall that Dijkstra's algorithm maintains a subset of processed vertices, maintains the minimum distance $d$ to an unprocessed vertex, and has the guarantee that $d(r,v)\le d$ for all processed vertices.

We alternate processing a vertex from the two Dijkstra's algorithms. Let $d_i$ be the minimum distance to an unprocessed vertex in $H_i$, where $i\in\{1,2\}$. We terminate immediately on the condition $d_1+d_2\ge0$, and output $V_r^{out}$ and $V_r^{in}$ as the processed vertices on the first and second graph, respectively. Suppose that the last graph processed was $H_i$. The algorithm sets $\Delta_r=d_2$ if $i=1$ and $\Delta_r=-d_1$ if $i=2$.

Suppose that the last graph processed was $H_1$. Let $d'_1<d_1$ be the \emph{maximum} distance to a \emph{processed} vertex on the first graph. Since the algorithm did not terminate before processing that vertex, we have $d'_1+d_2<0$. We now verify
 \begin{enumerate}
 \item All vertices $v\in V_r^{out}$ satisfy $d^-(r,v)=d_{H_1}(r,v)\le d'_1<-d_2=-\Delta_r$ by definition of $d'_1$. This also fulfills condition~(\ref{item:dijkstra-both-ways-5}).
 \item All vertices $v\notin V_r^{out}$ satisfy $d^-(r,v)=d_{H_1}(r,v)\ge d_1\ge-d_2=-\Delta_r$ by definition of $d_1$.
 \item All vertices $v\in V_r^{in}$ satisfy $d^0(v,r)=d_{H_2}(r,v)\le d_2=-\Delta_r$ as guaranteed of Dijkstra.
 \item All vertices $v\notin V_r^{in}$ satisfy $d^0(v,r)=d_{H_2}(r,v)\ge d_2=-\Delta_r$ as guaranteed of Dijkstra.
 \end{enumerate}
The case when the last graph processed was $H_2$ is similar. Let $d'_2<d_2$ be the maximum distance to a processed vertex on the second graph. Since the algorithm did not terminate before processing that vertex, we have $d_1+d'_2<0$. We now verify
 \begin{enumerate}
 \item All vertices $v\in V_r^{in}$ satisfy $d^0(v,r)=d_{H_2}(r,v)\le d'_2<-d_1=-\Delta_r$ by definition of $d'_2$. This also fulfills condition~(\ref{item:dijkstra-both-ways-5}).
 \item All vertices $v\notin V_r^{in}$ satisfy $d^0(v,r)=d_{H_2}(r,v)\ge d_2\ge-d_1=-\Delta_r$ by definition of $d_2$.
 \item All vertices $v\in V_r^{out}$ satisfy $d^-(r,v)=d_{H_1}(r,v)\le d_1=-\Delta_r$ as guaranteed of Dijkstra.
 \item All vertices $v\notin V_r^{out}$ satisfy $d^-(r,v)=d_{H_1}(r,v)\ge d_1=-\Delta_r$ as guaranteed of Dijkstra.
 \end{enumerate}
Finally, we discuss the running time. Naively, Dijkstra's algorithm on the first graph runs in $O(|V_r^{out}|n+|V_r^{out}|\log n)$ time using a Fibonacci heap: each of the $|V_r^{out}|$ processed vertices inserts its outgoing neighbors into the heap, and we pop from the heap once for each processed vertex. We can speed up this process to $O(|V_r^{out}|^2+|V_r^{out}|\log n)$ time as follows: instead of inserting all outgoing neighbors immediately after processing a vertex, perform the insertions lazily. That is, only insert the unprocessed neighbor with the smallest outgoing edge, and once that neighbor is processed, insert the next one. Here, we assume that each vertex has its outgoing edges initially sorted by weight. With this speedup, only $O(|V_r^{out}|^2)$ edges are ever checked by Dijkstra's algorithm, and the running time follows. The analysis on the second graph is identical.
\end{proof}

The algorithm runs \Cref{lem:dijkstra-both-ways} for \emph{all} negative vertices $r$. We now show that the running time is bounded, assuming the betweenness reduction guarantee.
\begin{lemma}\label{lem:total-size-bound}
Under the betweenness reduction guarantee of \Cref{lem:betweenness-reduction}, we have $\sum_{r\in V}(|V_r^{out}|+|V_r^{in}|)^2\le kn^2/b$ and $\sum_{r\in V}(|V_r^{out}|+|V_r^{in}|)\le O(kn/\sqrt b)$.
\end{lemma}
\begin{proof}
Let $N\subseteq V$ be the set of negative vertices, and consider a negative vertex $r\in N$. By the guarantees of \Cref{lem:dijkstra-both-ways}, we have $d^-(r,t)\le-\Delta_r$ for all $t\in V_r^{out}$ and $d^0(s,r)\le\Delta_r$ for all $s\in V_r^{in}$, which means that $d^1(s,r)+d^1(r,t)\le d^0(s,r)+d^-(r,t)\le0$. In fact, condition~(\ref{item:dijkstra-both-ways-5}) of \Cref{lem:dijkstra-both-ways} ensures that the inequality $\le0$ is actually strict. That is, $r$ is between $(s,t)$ for all $s\in V_r^{in}$ and $t\in V_r^{out}$. We \emph{charge} the vertex $r$ to all $|V_r^{in}|\cdot|V_r^{out}|$ such pairs.

By the betweenness reduction guarantee of \Cref{lem:betweenness-reduction}, each pair $(s,t)$ is charged at most $k/b$ times, so the total number of charges is at most $kn^2/b$. We conclude that
\[ \sum_{r\in N}|V_r^{in}|\cdot|V_r^{out}|\le\frac{kn^2}b .\]
Since $\big||V_r^{out}|-|V_r^{in}|\big|\le1$ from condition~(\ref{item:dijkstra-both-ways-6}) of \Cref{lem:dijkstra-both-ways}, we have $|V_r^{out}|+|V_r^{in}|\le O(\sqrt{|V_r^{in}|\cdot|V_r^{out}|+1})$, where the $+1$ handles the edge case $|V_r^{in}|=0$ or $|V_r^{out}|=0$. Therefore,
\[ \sum_{r\in N}(|V_r^{out}|+|V_r^{in}|)^2\le\sum_{r\in N}O(\sqrt{|V_r^{in}|\cdot|V_r^{out}|+1})^2\le \sum_{r\in N}O(|V_r^{in}|\cdot|V_r^{out}|+1)\le O\left(\frac{kn^2}b\right) .\]
Dividing by $k$ and applying Jensen's inequality,
\[ \left(\frac1k\sum_{r\in N}(|V_r^{out}|+|V_r^{in}|)\right)^2\le\frac1k\sum_{r\in N}(|V_r^{out}|+|V_r^{in}|)^2\le O\left(\frac{n^2}b\right)  ,\]
and rearranging finishes the proof.
\end{proof}

By \Cref{lem:dijkstra-both-ways,lem:total-size-bound}, we can compute sets $V_r^{out}$ and $V_r^{in}$ for all $r\in V$ in time $O(kn^2/b+kn\log n/\sqrt b)$. Next, the algorithm \emph{shortcuts} the graph by adding \emph{Steiner vertices} with the following edges to the graph:
 \begin{enumerate}
 \item For each negative vertex $r\in V$, create a new vertex $\tilde r$.
 \item For each $r\in V$, $u\in V_r^{out}\cup\{r\}$, and $v\in N^{out}(u)$, add the edge $(\tilde r,v)$ of weight $d^-(r,u)+\Delta_r+w(u,v)$ \emph{only if it is non-negative},
 \item For each $r\in V$, $v\in V_r^{in}\cup\{r\}$, and $u\in N^{in}(v)$, add the edge $(u,\tilde r)$ of weight $w(u,v)+d^0(v,r)-\Delta_r$ \emph{only if it is non-negative},
 \item For each negative vertex $r\in V$, $u\in V_r^{out}\cup\{r\}$, and $v\in N^{out}(u)$, add the edge $(r,v)$ of weight $d^-(r,u)+w(u,v)$, and\label{item:shortcut-4}
 \item For each negative edge $(r,r')$, $v\in V_r^{in}\cup\{r\}$, and $u\in N^{in}(v)$, add the edge $(u,r')$ of weight $w(u,v)+d^0(v,r)+w(r,r')$ \emph{only if $u$ is a negative vertex}.\label{item:shortcut-5}
 \end{enumerate}
For step~(\ref{item:shortcut-5}) above, recall that each vertex $r$ has at most one negative outgoing edge $(r,r')$. The total number of edges added is at most $2\sum_{r\in V}(|V_r^{out}|+1+|V_r^{in}|+1)n=O(kn^2/\sqrt b)$. In particular, the running time of this step is $O(kn^2/\sqrt b)$. We now claim that the added edges do not decrease distances in $G$. For any $r\in V$, note that each edge $(u,\tilde r)$ has weight at least $d(u,r)-\Delta_r$, and each added edge $(\tilde r,v)$ has weight at least $d(r,v)+\Delta_r$, so any $u\to \tilde r\to v$ path through $\tilde r$ has weight at least $d(u,r)+d(r,v)\ge d(u,v)$, as promised.

Let $G'$ be the resulting graph, where all added negative edges are considered new negative edges. Property~(\ref{item:single-step-3}) of \Cref{thm:single-step} holds by construction, and we have shown above that property~(\ref{item:single-step-1}) holds. We now show property~(\ref{item:single-step-2}), that adding these edges sufficiently reduces the number of negative edges along shortest $h$-negative hop paths.

\begin{lemma}\label{lem:shortcut}
Consider any $s,t\in V$ and an $(s,t)$-path $P$ with $h$ negative edges. After this shortcutting step, there is an $(s,t)$-path with at most $h-\lfloor h/3\rfloor$ negative edges, and whose weight is at most that of $P$.
\end{lemma}
\begin{proof}
Consider the negative edges of $P$ in the order from $s$ to $t$. Create $\lfloor h/3\rfloor$ disjoint groups of three consecutive negative edges; we wish to shortcut each group using at most two negative edges, so that the shortcut path has at most $h-\lfloor h/3\rfloor$ negative edges.

Consider a group of three consecutive negative edges $(u,u')$, $(r,r')$, and $(v,v')$, in that order. Recall that \Cref{lem:dijkstra-both-ways} computes a value $\Delta_r$ for vertex $r$. There are three cases; see \Cref{fig:shortcut2} for a visual reference.
\begin{figure}\centering
\includegraphics[scale=.65]{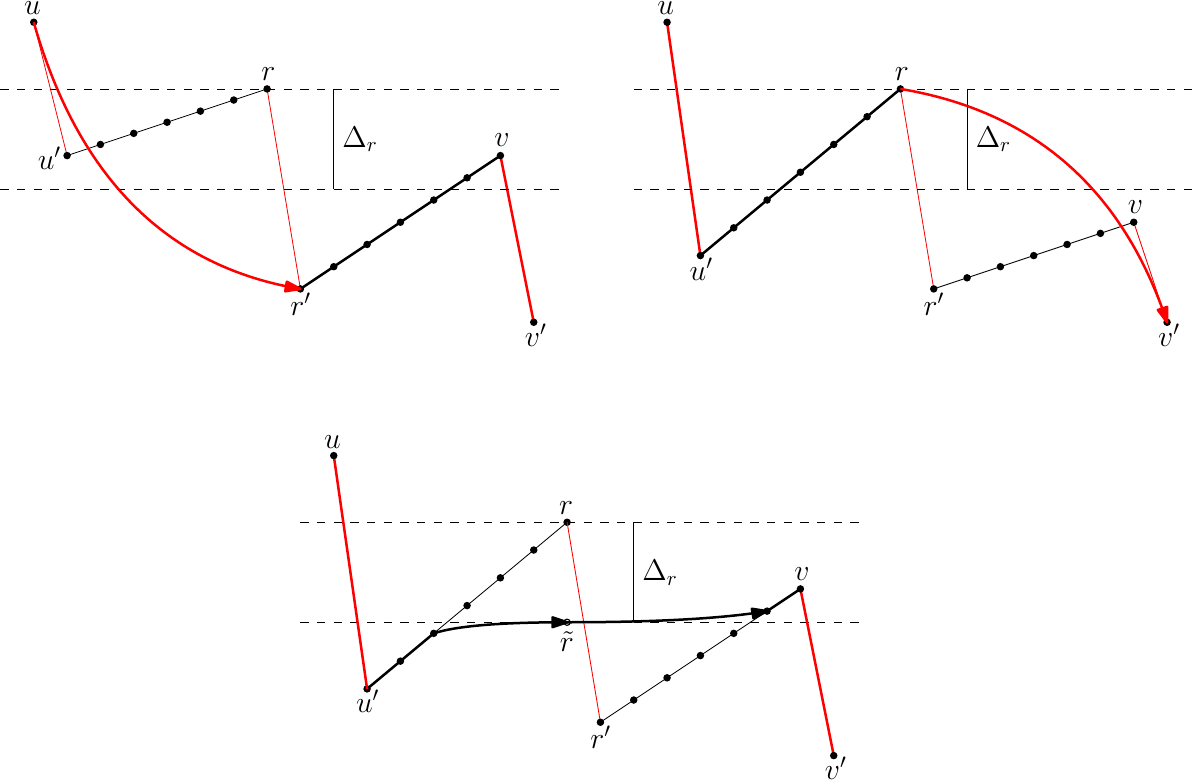}
\caption{The three cases of shortcutting in the proof of \Cref{lem:shortcut}, where the height indicates cumulative distance to each vertex. Negative edges are marked in red, and the new shortcut path is marked in bold.}\label{fig:shortcut2}
\end{figure}
 \begin{enumerate}
 \item $d^0(u',r)<\Delta_r$. In this case, condition~(\ref{item:dijkstra-both-ways-4}) of \Cref{lem:dijkstra-both-ways} ensures that $u'\in V_r^{in}$, so step~(\ref{item:shortcut-5}) adds the edge $(u,r')$ of weight $w(u,u')+d^0(u',r)+d(r,r')$, which is at most the weight of the $u\leadsto r'$ segment of $P$ since the negative edges $(u,u')$ and $(r,r')$ are consecutive along $P$. Therefore, we can shortcut the $u\leadsto r'$ segment of $P$, which has two negative edges, by the single (negative) edge $(u,r')$.
 \item $d^-(r,v)<-\Delta_r$. In this case, condition~(\ref{item:dijkstra-both-ways-2}) of \Cref{lem:dijkstra-both-ways} ensures that $v\in V_r^{out}$, so step~(\ref{item:shortcut-4}) adds the edge $(r,v')$ of weight $d^-(r,v)+w(v,v')$, which is at most the weight of the $r\leadsto v'$ segment of $P$ since the negative edges $(r,r')$ and $(v,v')$ are consecutive along $P$. Therefore, we can shortcut the $r\leadsto v'$ segment of $P$, which has two negative edges, by the single (negative) edge $(r,v')$. 
 \item $d^0(u',r)\ge\Delta_r$ and $d^-(r,v)\ge-\Delta_r$. Consider the $u'\leadsto r$ segment of $P$, and let $u''$ be the last vertex on this segment with $d^0(u'',r)\ge\Delta_r$. Then, either $u''=r$ or the vertex after $u''$ is in $V_r^{in}$, and in both cases, the edge $(u'',\tilde r)$ is added with weight $d^0(u'',r)-\Delta_r\ge0$. Similarly, consider the $r'\leadsto v$ segment of $P$, and let $v''$ be the first vertex on this segment with $d^-(r,v'')\ge-\Delta_r$. Then, either $v''=r'$ or the vertex before $v''$ is in $V_r^{out}$, and in both cases, the edge $(\tilde r,v'')$ is added with weight $d^-(r,v'')+\Delta_r\ge0$. The $u''\to\tilde r\to v''$ path consists of two non-negative edges and has weight $d^0(u'',r)+d^-(r,v'')$, which is at most the weight of the $u''\leadsto r\leadsto v''$ segment of $P$, so we can shortcut this segment of $P$, which contains the negative edge $(r,r')$, by the two-hop path $u''\to\tilde r\to v''$ with no negative edges.\qedhere
 \end{enumerate}
\end{proof}
With the three properties established, this concludes the proof of \Cref{thm:single-step}.

\subsection{Solving the Recurrence}\label{sec:recurrence}

Finally, we solve the running time recurrence from \Cref{cor:recurrence} to $O(\sqrt k(n+k\log k)^2\log^{2.5}n)$. Let constant $C>0$ be large enough that all of the hidden constants in the statement of \Cref{thm:single-step} can be replaced by $C$. First, if $k\le4C^3\log^3n$, then naively computing $k$-negative hop shortest paths suffices, since the running time $O(kn^2)$ meets the desired bound. Otherwise, we apply \Cref{thm:single-step} with $b=\frac k{4C^3\log^3n}$, so that the algorithm makes at most $C\log k$ recursive calls on graphs with at most $n+Ck\log k$ vertices and $Cb\log n=\frac k{4C^2\log^2n}$ negative vertices, and runs in at most $Ck\log n\cdot(n+Ck\log k)^2/\sqrt b=2C^{2.5}\sqrt k(n+Ck\log k)^2\log^{2.5}n$ additional time. On recursion level $i\ge1$, there are at most $(C\log n)^i$ recursive instances, each with at most $n+iCk\log k\le i(n+Ck\log k)$ vertices and running time at most
\[ 2C^{2.5}\sqrt{\frac k{(4C^2\log^2n)^i}}(i(n+Ck\log k))^2\log^{2.5}n=\frac{2i^2C^{2.5}\sqrt k(n+Ck\log k)^2\log^{2.5}n}{(2C\log n)^i}, \]
so the total running time at level $i$ is at most
\[ (C\log n)^i\cdot\frac{2i^2C^{2.5}\sqrt k(n+Ck\log k)^2\log^{2.5}n}{(2C\log n)^i}=\frac{i^2}{2^i}\cdot2C^{2.5}\sqrt k(n+Ck\log k)^2\log^{2.5}n. \]
Since the series $i^2/2^i$ converges, the total running time is $O(\sqrt k(n+k\log k)^2\log^{2.5}n)$. Finally, \Cref{thm:main} follows from the trivial bound $k\le n$.

\section{Reducing Betweenness Using Shortest Path}

In this section, we generalize our recursive betweenness reduction to prove \Cref{thm:main-sparse}. We first show that betweenness reduction, which is the bottleneck of previous algorithms, can be reduced to the SSSP problem on a larger graph with fewer negative edges.

\begin{lemma}\label{lem:reduction_bw_to_sssp}
    There is an algorithm that, given a graph $G$ and two parameters $h,b$, either declares that $G$ has a negative cycle or produces a valid potential $\phi$ such that with high probability,
    \[\left|\{x\in V\mid d^h_\phi(u,x)+d^h_\phi(x,v)<0\}\right|\leq n/b\quad\forall u,v\in V.\]
    The algorithm uses one oracle call to a SSSP instance with $O(hm)$ edges and $O(b\log n)$ negative edges, and $O(hm)$ additional time.
\end{lemma}

We first review the previously used $\tilde O(mhb+nb^2)$ betweenness reduction algorithm from \cite{fineman2024single}, which is similar to \Cref{lem:betweenness-reduction}, except that $h$-hop distances are used here. The algorithm first samples a set $S$ of $\Theta(b\log n)$ random vertices. Then, for each vertex $w\in S$, they compute the $h$-hop shortest paths $d^h(\cdot,
w)$ and $d^h(w,\cdot)$, that is, the distances from and to $w$, which takes $\tilde O(mhb)$ time. Next, the algorithm constructs an auxiliary graph that contains a copy of $G$'s vertices and $|S|$ additional vertices corresponding to the vertices in $S$. For each vertex $w\in S$, we add an additional vertex $w'$ in this graph, and for each $u\in V$, we add an edge from $u$ to $w'$ with weight $d^h(u,w)$, and an edge from $w'$ to $u$ with weight $d^h(w,u)$. Finally, we compute shortest path distances on this auxiliary graph, which contains $O(nh\log n)$ edges and $O(h\log n)$ negative vertices, in $\tilde O(nh^2)$ time, and we apply these distances as a potential function to the original graph.

The correctness analysis of the algorithm above is the same as the argument in \Cref{lem:betweenness-reduction}. For every pair $s,t\in V$ and any sampled $w\in S$, we have $d^h_\phi(s,w)+d^h_\phi(w,t)\geq 0$ because of the path $s\to w\to t$ in the auxiliary graph. This inequality is equivalent to $\phi(t)-\phi(s)\leq d^h(s,w)+d^h(w,t)$. Moreover, a vertex $u$ is in the $h$-hop betweenness only if $\phi(t)-\phi(s)>d^h(s,u)+d^h(u,t)$, so $u$ can be in the betweenness only if $d^h(s,u)+d^h(u,t)$ is smaller than the corresponding value for any sampled $w$. Therefore, the betweenness is at most $n/b$ with high probability. 

We now introduce our new ideas. The $O(mhb)$ bottleneck in the algorithm above arises from computing the $h$-hop shortest path from $S$ in both directions. Instead of computing these distances explicitly, we can simulate them by inserting a layered graph into the auxiliary graph. Suppose we create $h+1$ copies of the non-negative graph $G^+$ and insert the negative edges in $G$ that link one copy to the next. Then, the distance of a vertex $u$ in the first copy to a vertex $v$ on the last copy equals $d^h(u,v)$. Therefore, we can replace $d^h(w,\cdot)$ and $d^h(\cdot,w)$ with two layered graphs and then compute shortest path distances on the enlarged auxiliary graph to complete the betweenness reduction. The only issue to overcome is that the negative edges between layers introduce additional negative edges. We address this by adjusting the weights between layers so that only the edges from one specific layer to the next are negative.

\newcommand{\tforw}[1]{\textrm{forward}#1}
\newcommand{\tback}[1]{\textrm{backward}#1}

\begin{proof}[Proof of \Cref{lem:reduction_bw_to_sssp}]
    The shortest path instance we construct contains $2h+1$ copies of $G^+$, where $G^+$ denotes the graph $G$ with all negative edges removed. These copies are denoted by $G_{0}$, $G^{\tforw{}}_{{1}},\ldots,G^{\tforw{}}_{{h}}$, and $G^{\tback{}}_{{1}},\ldots,G^{\tback{}}_{{h}}$. The vertices in each copy are referenced using the vertex name with the subscript of its corresponding graph copy, such as $u_{0}$ and $u^{\tforw{}}_{{1}}$.

    Let $M$ be a value larger than the absolute value of all negative edge weights. The edges between these copies are constructed in the following steps. First, for each negative edge from $u$ to $v$ in $G$ with weight $w$, add the following set of edges, each with weight $w+M\ge0$.

    \begin{itemize}
        \item An edge from $u_{0}$ to $v^{\tforw{}}_{{1}}$.
        \item An edge from $u^{\tforw{}}_{{1}}$ to $v^{\tforw{}}_{{2}}$, and so on, until from $u^{\tforw{}}_{{h-1}}$ to $v^{\tforw{}}_{{h}}$.
        \item An edge from $u^{\tback{}}_{{h}}$ to $v^{\tback{}}_{{h-1}}$, and so on, until from $u^{\tback{}}_{{2}}$ to $v^{\tback{}}_{{1}}$.
        \item An edge from $u^{\tback{}}_{{1}}$ to $v_{0}$.
    \end{itemize}
    Next, sample a random set $S$ with $O(b\log n)$ vertices; for each vertex $w\in S$, add an edge from $w^{\tforw{}}_{{h}}$ to $w^{\tback{}}_{{h}}$, with weight $-2hM$.

    Now, we show that computing shortest path distances on this graph suffices for betweenness reduction. First, suppose that we find a negative cycle on this graph. Since the edges in this graph correspond to edges in $G$, this cycle corresponds to a closed walk in $G$, with the only difference being that the edges between copies have additional weights. Moreover, the cycle must traverse the copies in the following order:
    \[G_{0}\to G^{\tforw{}}_{{1}}\to\ldots\to G^{\tforw{}}_{{h}}\to G^{\tback{}}_{{h}}\to\ldots\to G^{\tback{}}_{{1}}\to G_{0}.\] 
    For each such loop traversed, there are $2h$ edges with an additional $+W$ weight, and one edge (from $G^{\tforw{}}_{{h}}$ to $G^{\tback{}}_{{h}}$) with a $-2hW$ weight, so these additional weights are always canceled out. Therefore, the weight of a cycle in this graph equals the weight of the corresponding walk in $G$, implying that a negative cycle in this graph indicates a negative closed walk in $G$, and hence a negative cycle in $G$.

    Next, suppose that the shortest paths algorithm returns a distance function $\phi$. We claim that
    \[\phi(v_0)-\phi(u_0)\leq \min_{w\in S}d^h(u,w)+d^h(w,v)\quad\forall u,v\in V.\]
    The reason is that, for any vertex $w\in S$, the path achieving $d^h(u,w)+d^h(w,v)$ can be routed on our new graph in the following way: route the $h$-hop shortest path $d^h(u,w)$ from $u_0$ to $w^{\tforw{}}_{{h}}$, traverse the edge from $w^{\tforw{}}_{{h}}$ to $w^{\tback{}}_{{h}}$, and then route the $h$-hop shortest path $d^h(w,v)$ from $u^{\tback{}}_{{h}}$ to $w_{0}$. Since $\phi$ is a distance function, any path from $u_{0}$ to $v_{0}$ has weight at least $\phi(v_{0})-\phi(u_{0})$, proving the claim.
\end{proof}

To solve this shortest path instance, we apply the algorithm in \cite{huang2026faster}. However, we now want the time complexity to depend on $k$, the number of negative edges, rather than of the number of vertices, because our auxiliary graph in \Cref{lem:reduction_bw_to_sssp} contains fewer negative edges. Fortunately, the analysis in \cite{huang2026faster} mostly remains valid when $k<n$, except that the $O(nh^2)$ factor in betweenness reduction becomes a bottleneck even for denser graphs. Therefore, the technique used in Section~6 to eliminate this factor for sparse graphs should be applied in all cases.

\begin{theorem}[Implicit in \cite{huang2026faster}]\footnote{By combining Lemma 6.1 with Lemma 4.5 in the paper.}\label{thm:hjq26_sssp}
    There is an algorithm that solves the real-weight SSSP problem with $m$ edges and $k$ negative edges in $\tilde O(k^2+mk^{3/4})$ time with high probability.
\end{theorem}

An important observation is that the shortest path instance produced by \Cref{lem:reduction_bw_to_sssp} has $O(mh)$ edges and only $O(b\log n)$ negative edges. For most common choices of the parameters $h$ and $b$, the number of negative edges is much smaller than the graph size, and the $mk^{3/4}$ term dominates, yielding a running time of $\tilde O(mhb^{3/4})$.

\begin{theorem}\label{thm:faster_bw_reduction}
    There is an algorithm that, given a graph $G$ and two parameters $h,b$ that $mh\geq b^{5/4}$, in $\tilde O(mhb^{3/4})$ time, either declares that $G$ has a negative cycle, or produces a valid potential $\phi$, such that with high probability, under this reweighting,
    \[\left|\{x\mid d^h_\phi(u,x)+d^h_\phi(x,v)<0\}\right|\leq n/b\quad\forall u,v\in V.\]
\end{theorem}

\subsection{Improving the betweenness reduction in \texorpdfstring{\cite{huang2025faster}}{}}

Because computing the betweenness reduction is the bottleneck for previous works, solving it faster via \Cref{thm:faster_bw_reduction} leads to an improved running time. To illustrate this, we show how using \Cref{thm:faster_bw_reduction} can improve the algorithm in \cite{huang2025faster}, which results in the fastest algorithm for SSSP when $m=O(n^{5/4})$. The improved algorithm uses \Cref{thm:faster_bw_reduction} for betweenness reduction, a different set of parameters $h,b$ chosen for balancing the complexity, and leaves the other parts unchanged.

\begin{theorem}\label{thm:improved_sssp}
    There is an algorithm that solves the real-weight SSSP problem with $m$ edges and $k$ negative edges in $\tilde O(mk^{7/9})$ time with high probability.
\end{theorem}

To prove \Cref{thm:improved_sssp}, we first review their algorithm. They follow the approach of~\cite{fineman2024single}, where one iteratively finds a set of negative vertices $U$ to neutralize. The goal is find $U$ such that we can compute Johnson's reweighting for $G_U$ ($G$ with negative edges not incident to $U$ removed) using $m|U|/k^\epsilon$ time. This way, iteratively neutralizing all the edges takes $mn^{1-\epsilon}$ time. Towards finding such $U$,~\cite{fineman2024single} introduced many concepts, which we summarize below.

\paragraph{Hop reducers and remote sets.} Naively computing Johnson's reweighting for $G_U$ naively takes $\tilde{O}(m|U|)$ time, which is too slow. To improve upon this,~\cite{fineman2024single} introduced auxiliary hop reduction graphs to decrease the number of hops in a shortest path. Indeed, given a $h$-hop reducer for $G_U$, we would be able to compute Johnson's reweighting for $G_U$ in $\tilde{O}(m|U|/h)$ time.

\begin{definition}
    Let $h\in\mathbb{N}$. We say that a graph $H=(V_H,E_H)$ is an $h$-hop reducer for a graph $G=(V,E)$ if $V\subseteq V_H$ and $d_G(s,t)\le d_H^{\lceil k/h\rceil}(s,t)\le d^k_G(s,t)$ for all $s,t\in V$ and $k\in\mathbb{N}$.
\end{definition}

In order to construct an $h$-hop reducer of smaller size,~\cite{fineman2024single} introduced the notion of ``remote'' sets of negative vertices. We say that a vertex $s\in V$ can $h$-hop negatively reach a vertex $t\in V$ if there is an $h$-hop path from $s$ to $t$ with negative length. For a parameter $r\in\mathbb{N}$ and a set of negative vertices $U\subseteq V$, we say that $U$ is $r$-remote if the union of the $r$-hop negative reaches of vertices in $U$ has at most $n/r$ vertices. They show that remote sets have linear-size hop reducers:

\begin{lemma}[\cite{fineman2024single}, Lemma 3.3]\label{lem:hop-reducer}
   Let $G$ have maximum in-degree and out-degree $O(m/n)$. Given a set of $r$-remote vertices $U$, one can construct an $r$-hop reducer for $G_U$ with $O(m)$ edges and $O(n)$ vertices in $\tilde{O}(m)$ time. 
\end{lemma}

\paragraph{Betweenness reduction and negative sandwiches.}
In order to construct a remote set,~\cite{fineman2024single} introduce the notion of a negative sandwich, which is a triple $(s,U,t)$, where $s,t\in V$ and $U\subseteq N$, such that $d^1(s,u)<0$ and $d^1(u,t)<0$ for all $u\in U$.~\cite{huang2025faster} generalized this to a weak $h$-hop negative sandwich, which is a triple $(s,U,t)$ such that $d^h(s,u)+d^h(u,t)<0$ for all $u\in U$. They showed that if the betweenness is small, a weak sandwich can be converted into a remote set.

\begin{lemma}[\cite{huang2025faster}, Lemma 2.3]\label{lem:sandwich-to-remote}
    Let $(s,U,t)$ be a weak $h$-hop negative sandwich and let $(s,t)$ have $(r+h)$-hop betweenness $n/r$. Then one can compute valid potentials $\phi$ making $U$ r-remote in $\tilde{O}(m\cdot(r+h))$ time.
\end{lemma}

To guarantee small betweenness, they used the betweenness reduction algorithm from~\cite{fineman2024single}.
\begin{lemma}[\cite{fineman2024single}, Lemma 3.5]\label{lem:betweenness-reduction-slow}
    There is an $\tilde{O}(mhb)$ time randomized algorithm that returns either a set of valid potentials $\phi$ or a negative cycle. With high probability, all pairs $(s,t)\in V\times V$ have $O(h)$-hop betweenness at most $n/b$ in the graph $G_{\phi}$.
\end{lemma}

The final step is a way to find a large negative sandwich, which we will neutralize.

\begin{lemma}[\cite{huang2025faster}, Lemma 4.2]\label{lem:big-sandwich}
    For $h=\Omega(\log{n})$, there is a randomized algorithm running in $\tilde{O}(mh)$ time that returns either a negative cycle, a weak $h$-hop sandwich $(s,U,t)$, or a set of negative vertices $S$ and the distances $d_S(V,v)$ for all vertices $v\in V$. With high probability, we have $|U|,|S|\ge\Omega(\sqrt{hk})$ (when they are returned by the algorithm).
\end{lemma}


Finally, we sketch the proof of the $\tilde{O}(mn^{4/5})$ time algorithm, where we assume for simplicity that there is no negative cycle. We then show that using the faster betweenness reduction (\Cref{thm:faster_bw_reduction}), we can do slightly better using the same exact framework.

\paragraph{$\mathbf{\tilde{O}(mn^{4/5})}$ time algorithm.} Let $h=b=k^{1/5}$. We first reduce the $h$-hop betweenness to $n/b$ in time $\tilde{O}(mhb)$ (\Cref{lem:betweenness-reduction-slow}). We then either directly neutralize a set $S$ of size $\tilde{\Omega}(\sqrt{hk})$ or find a weak $h$-hop negative sandwich $(s,U,t)$ of size $\tilde{\Omega}(\sqrt{hk})$, where $\sqrt{hk}=k^{3/5}$ (\Cref{lem:big-sandwich}). If we find a negative sandwich, we reweight $G_U$ to make $U$ $h$-remote in time $\tilde{O}(mh)$ (\Cref{lem:sandwich-to-remote}). We then construct an $h$-hop reducer $H$ of $G_U$ with $O(m)$ edges and $O(n)$ vertices in time $\tilde{O}(m)$ (\Cref{lem:hop-reducer}). Using $H$, we can compute Johnson's potentials for $G_U$ in $\tilde{O}(m|U|/h)$ time. Overall, we spend $\tilde{O}(mk^{2/5})$ time to neutralize $\tilde{\Omega}(k^{3/5})$ negative vertices. Iterating this process yields an $\tilde{O}(mk^{4/5})$ time algorithm for neutralizing all negative vertices.

\paragraph{Proof of \Cref{thm:improved_sssp}.} Let $h=b=k^{2/9}$. We again reduce the $h$-hop betweenness to $n/h$, but now in time $\tilde{O}(mhb^{3/4})=\tilde{O}(mk^{7/18})$ (\Cref{thm:faster_bw_reduction}). Note that we can apply~\Cref{thm:faster_bw_reduction} because $mh\ge b^{5/4}$, as required. We then either directly neutralize a set $S$ of size $\tilde{\Omega}(\sqrt{hk})$ or find a weak $h$-hop negative sandwich $(s,U,t)$ of size $\tilde{\Omega}(\sqrt{hk})$, where $\sqrt{hk}=k^{11/18}$ (\Cref{lem:big-sandwich}). If we find a negative sandwich, we reweight $G_U$ to make $U$ $h$-remote in time $\tilde{O}(mh)$ (\Cref{lem:sandwich-to-remote}). We then construct an $h$-hop reducer $H$ of $G_U$ with $O(m)$ edges and $O(n)$ vertices in time $\tilde{O}(m)$ (\Cref{lem:hop-reducer}). Using $H$, we can compute Johnson's potentials for $G_U$ in $\tilde{O}(m|U|/h)=\tilde{O}(mk^{7/18})$ time. Overall, we spend $\tilde{O}(mk^{7/18})$ time to neutralize $\tilde{\Omega}(k^{11/18})$ negative vertices. Iterating this process yields an $\tilde{O}(mk^{7/9})$ time algorithm for neutralizing all negative vertices.

\section*{Acknowledgement}

George Li is supported by the National Science Foundation Graduate
Research Fellowship Program under Grant No.\ DGE2140739.

\bibliographystyle{alpha}
\bibliography{ref}

\appendix

\section{Refining \texorpdfstring{\cite{huang2026faster}}{}}


We first give a recap of the algorithm in \cite{huang2026faster}, whose approach follows the same structure as previous work. They first preprocess the graph using some reweighting which reduces the betweenness, so that the negative paths in the graph have more structure. They then identify a set of negative vertices $U$ in the graph, called a negative sandwich, which has additional structure making these negative edges easy to efficiently neutralize. Finally, they construct a hop-reducer to efficiently compute Johnson's potentials for the graph $G_U$, thus neutralizing the negative edges incident to $U$. We now explain this process in more detail.

\paragraph{Betweenness Reduction.}

Before choosing the set of vertices to neutralize, they first preprocess the graph using a reweighting so that there is additional structure, as defined below.

\begin{definition}
    For a parameter $h\in\mathbb{N}$ and $s,t\in V$, the $h$-hop betweenness of $s,t$ is the number of vertices $v$ such that $d^h(s,v)+d^h(v,t)<0$.
\end{definition}

\begin{lemma}[\cite{huang2026faster}, Lemma 4.1]\label{lem:betweenness-reduction-26}
    Given a graph $G$ and some parameter $h=\Omega(\log{n})$, one can compute potentials $\phi$ such that for all $s,t\in V$ and $\eta\in \mathbb{N}$, the $(\eta+O(\log{n}))$-hop betweennness of $s,t$ is at most $n\eta/h$ in $G_\phi$. The algorithm runs in $\tilde{O}(mh+nh^2)$ time.
\end{lemma}

\paragraph{Negative Sandwiches.}

Following previous work, they find a weak $h$-hop negative sandwich for $h\in\mathbb{N}$, which is a triple $(s,U,t)$, where $s,t\in V$ and $U$ is a subset of negative vertices, such that $d^h(s,u)+d^h(u,t)<0$ for all $u\in U$. The step of finding a large negative sandwich is done by using proper hop distances and edge sampling, as in the following lemma.

\begin{lemma}[\cite{huang2025faster}, Lemma 4.2]\label{lem:apx-big-sandwich}
    For $h=\Omega(\log{n})$, there is a randomized algorithm running in $\tilde{O}(mh)$ time that returns either a negative cycle, a weak $h$-hop sandwich $(s,U,t)$, or a set of negative vertices $S$ and the distances $d_S(V,v)$ for all vertices $v\in V$. With high probability, we have $|U|,|S|\ge\Omega(\sqrt{hk})$ (when they are returned by the algorithm).
\end{lemma}

The key property of negative sandwiches $(s,U,t)$ is that when combined with the guarantee that $(s,t)$ have low betweenness, we can show that the graph can be reweighted such that $U$ has small negative reach. This makes neutralizing the negative edges incident to $U$ significantly easier.

\begin{lemma}[\cite{huang2026faster}, Lemma 3.4]\label{lem:hjq26_bw_to_neg_reach}
    Let $(s,U,t)$ be a weak $h$-hop negative sandwich and let $(s,t)$ have $(k+h)$-hop betweennness $n/\ell$. Let $\phi(v)=\min\{d^h(s,v),-d^h(v,t)\}$. Then the $k$-hop negative reach of $U$ in $G_{\phi}$ has at most $n/\ell$ vertices.
\end{lemma}

\paragraph{Hop Reducers.}
We wish to compute a reweighting to neutralize all negative edges incident to $U$. Naively computing Johnson's reweighting for $G_U$ takes $\tilde{O}(m|U|)$ time, which is too slow. To improve upon this,~\cite{fineman2024single} introduced auxiliary hop reduction graphs to decrease the number of hops in a shortest path. Indeed, given a $h$-hop reducer for $G_U$, we would be able to compute Johnson's reweighting for $G_U$ in $\tilde{O}(m|U|/h)$ time.

\begin{definition}\label{def:hop-reducer}
    Let $h\in\mathbb{N}$. We say that a graph $H=(V_H,E_H)$ is an $h$-hop reducer for a graph $G=(V,E)$ if $V\subseteq V_H$ and $d_G(s,t)\le d_H^{\lceil k/h\rceil}(s,t)\le d^k_G(s,t)$ for all $s,t\in V$ and $k\in\mathbb{N}$.
\end{definition}

In order to compute the $h$-hop reducer, \cite{huang2026faster} iteratively builds $\eta$-hop reducers from $\eta/2$-hop reducers starting with the trivial $1$-hop reducer. To do this, they define a notion of distance estimates which we omit for brevity, and prove the following two lemmas, one which obtains distance estimates from $\eta/2$-hop reducers and one which uses these distance estimates to build $\eta$-hop reducers. 


\begin{lemma}[\cite{huang2026faster}, Lemma 4.3]\label{lem:construct-hop-reducer}
    Let $i\in[L]$. Suppose we have valid distance estimates $\delta_j(s,t)$ at level $j$ for all $j<i$, $s\in U$, and $t\in\overline{U}$. Then one can construct a $2^{i-1}$-hop reducer $H_i$ for $G_i$ with $n_i=O(2^in/h)$ vertices and $m_i=O(2^im/h+|U|^2i)$ edges in $O(m_i+n_i\log{n_i})$ time.
\end{lemma}

\begin{lemma}[\cite{huang2026faster}, Lemma 4.4]\label{lem:compute-distance-estimates}
    Given a $2^{i-1}$-hop reducer $H_i$ for $G_i$ with $m_i=O(2^im/h+|U|^2i)$ edges and $n_i=O(2^in/h)$ vertices, one can compute valid distance estimates $\delta_i(s,t)$ at level $i$ for all $s\in U$ and $t\in \overline{U}$ with high probability in $O\left(\min\{1,\log(n)/2^i\}(2^i|U|\mu/h+|U|^3i)\right)$ randomized time.
\end{lemma}

Now, we can sketch the proof of the $\tilde{O}(mk^{3/4}+nk)$ time algorithm, focusing on the cases which are bottlenecks for the running time. We first apply~\Cref{lem:betweenness-reduction-26} to reduce betweenness in $\tilde{O}(mh)$ time. Next, we apply~\Cref{lem:apx-big-sandwich} to find a weak $h_0$-hop negative sandwich $(s,U,t)$ of size $|U|=\tilde{\Omega}(\sqrt{k})$ for $h_0=O(\log{n})$. By~\Cref{lem:hjq26_bw_to_neg_reach}, we can compute a reweighting such that after the reweighting $U$ can negatively reach at most $n\eta/h$ vertices in $\eta$ hops for all $\eta\in\mathbb{N}$. With a slight abuse of notation, let $G$ denote the graph after the reweightings. The goal is to construct a $h$-hop reducer for the graph $G_U$.

To construct the $h$-hop reducer, we iteratively construct $2^i$-hop reducers from $2^{i-1}$-hop reducers, over $L=\lceil \log_2h\rceil +1$ levels. For each $i\in[L]$, let $V_i$ denote the $2^i$-hop negative reach of $U$ and let $G_i$ denote the induced subgraph of $G[V_i]$. Note that by the betweenness reduction guarantees, $G_i$ has $n_i=O(2^in/h)$ vertices and $m_i=O(2^im/h)$ edges. Starting with $H_1=G_1$ being a $2^{i-1}$-hop reducer for $G_1$ when $i=1$, we iteratively take the $2^{i-1}$-hop reducer for $G_{i-1}$ and apply \Cref{lem:compute-distance-estimates} to compute distance estimates $\delta_{i}(s,t)$. Then, we use these distance estimates to apply \Cref{lem:construct-hop-reducer} to construct a $2^{i}$-hop reducer for $G_{i}$. Iterating this process for $i\in[L]$ gives the desired $h$-hop reducer.

In summary, we use $\tilde O(mh+nh^2)$ time for the betweenness reduction and $\tilde{O}(mh)$ time to find a $O(\log n)$-hop negative sandwich $U$ of size $\tilde{\Omega}(\sqrt k)$ in \Cref{lem:betweenness-reduction-26,lem:apx-big-sandwich,lem:hjq26_bw_to_neg_reach}. We then use $\tilde{O}(m|U|/h+k^{3/2})$ time to construct the $h$-hop reducer, after which computing the Johnson's potentials to neutralize $U$ takes $\tilde{O}(m\sqrt{k}/h+k^{3/2})$ time. Setting $h=k^{1/4}$, the algorithm neutralizes $\tilde{\Omega}(\sqrt k)$ negative edges in $\tilde O(nk^{1/2}+mk^{1/4}+k^{3/2})$ time, so the running time for neutralizing all edges is $\tilde O(nk+mk^{3/4})$.





\subsection{Alternative Hop-Reducer Construction}

Our first observation is that we can also construct a hop reducer for a graph $G$ using one call to a SSSP oracle. We explain after how this can be used to improve~\cite{huang2026faster}.

\begin{claim}\label{cl:easy-hop-reducer}
    Let $G$ be a graph with $k$ negative vertices. We can compute a $k$-hop reducer $H$ of $G$ with $O(m)$ edges using one call to an SSSP oracle on graph $G$ and $\tilde{O}(m)$ additional time.
\end{claim}
\begin{proof}
    Using the SSSP oracle, we can compute Johnson's reweighting $\phi(v)=d(V,v)$ for each $v\in V$.
    Using this $\phi$, we can define our hop-reducer $H$. Initialize $H$ with two copies of the graph $G$, denoted $G_1=(V_1,E_1)$ and $G_2=(V_2,E_2)$. We set the weight of edges in $G_1$ based on the original graph $G$ and the weight of edges in $G_2$ based on the neutralized graph $G_{\phi}$. Next, we define the edges between $G_1$ and $G_2$. For a vertex $v\in V$, let $v_1$ and $v_2$ denote the copies of $v$ in $G_1$ and $G_2$, respectively. Also, let $\phi_{\max}=\max_{v\in V}\phi(v)$. For every vertex $v\in V$, we add an edge from $v_1$ to $v_2$ with weight $\phi_{\max}-\phi(v)$, and an edge from $v_2$ to $v_1$ with weight $\phi(v)-\phi_{\max}$. 
    
    To prove that $H$ is a $k$-hop reducer of $G$, it suffices by \Cref{def:hop-reducer} to prove $d_G(s,t)=d_H^1(s,t)$ for each $s,t\in V$. Consider any $s$-$t$ path $P$ from $s$ to $t$ in $G$; we claim there is a $s_1$-$t_1$ path in $H$ which contains one negative edge and has the same length as $P$. Take the edge $(s_1,s_2)$ which has weight $\phi_{\max}-\phi(s)$, the path following (the copy of) $P$ in $G_2$ which has weight $w(P)+\phi(s)-\phi(t)$, and finally the edge $(t_2,t_1)$ which has weight $\phi(t)-\phi_{\max}$. This $s_1$-$t_1$ path has total weight $w(P)$, and the only negative edge is the $(t_2,t_1)$ edge. Next, consider any 1-hop $s_1$-$t_1$ path in $H$; we claim there is a $s$-$t$ path in $G$ which has the same length as $P$. If the path had 0 hops, the path would be in $G_1$, so the claim is trivial. If the path had 1 hop, it is without loss of generality of the form: take the edge $(s_1,s_2)$, a $s_2$-$t_2$ path $P$ in $G_2$, and finally the edge $(t_2,t_1)$. By the same calculation as before, this path has the same weight as if we followed the path $P$ in $G$.
\end{proof}

Now, we explain why this naive hop-reducer construction is useful. At level $i\in[L]$ in the bootstrapping process, the graph $G_i$ has $O(m/(h/2^i))$ edges and $|U|$ negative edges. In~\cite{huang2026faster}, it takes $\tilde{O}(m|U|/h+|U|^3/2^i)$ time to compute the $2^i$-hop reducer. Suppose that we recursively use an $\tilde O(mk^\alpha)$ time algorithm for computing the shortest path, for some $\alpha<1$, when we apply \Cref{cl:easy-hop-reducer}. Then computing the $2^i$-hop reducer would take $\tilde{O}(m|U|^{\alpha}2^i/h)$ time, which is dominated by $\tilde{O}(m|U|/h)$ whenever $2^i\le |U|^{1-\alpha}$. This allows us to do bootstrapping starting from this level and skipping the lower bootstrapping levels. 

The main benefit of skipping the earlier levels is that we can use a slightly weaker betweenness reduction guarantee. Specifically, we only need to guarantee that the $(\eta+2^j)$-hop negative reach for $\eta\in[2^j,h]$ is small, where $2^j=|U|^{1-\alpha}$. This allows us to pick a larger $h_0$ at the beginning, because we only need to bound the negative reach $V_i$ for larger $i$. Since we only want to bound the size of the $\eta$-hop negative reach for $\eta\in[2^j,h]$, then the algorithm in \Cref{lem:betweenness-reduction-26} can also be used to ensure that for any $\eta\in[2^j,h]$, the $(\eta+2^j)$-hop betweenness is at most $n/(h/\eta)$, without changing the running time. This enables us to use a larger $h=2^j$ when finding the weak $h$-hop negative sandwich, thus finding larger sandwiches since the size in \Cref{lem:apx-big-sandwich} is $|U|=\tilde{\Omega}(\sqrt{kh})$.

\subsection{Recursive (Weak) Betweenness Reduction}

Secondly, we use the idea in~\Cref{thm:faster_bw_reduction} to reduce the weaker betweenness reduction guarantees in \Cref{lem:betweenness-reduction-26} to a shortest path problem. We note that this cannot directly improve the running time of the betweenness reduction routine, because the weaker betweenness reduction still requires $h$-hop SSSP which takes $\tilde{O}(mh)$ time and the running time of \Cref{lem:betweenness-reduction-26} is already $\tilde{O}(mh)$. Even so, we show how the idea in~\Cref{thm:faster_bw_reduction} can further improve our algorithm.


Instead of improving the runtime directly, we observe that the reduction helps us obtain a better betweenness guarantee for smaller $\eta$ in the same running time. This, in turn, means that the graph $V_i$ used for hop reducer construction has fewer edges in the first few levels. Combined with the idea in the previous subsection, the recursive betweenness reduction allows us to skip more levels in the same running time, which leads to even larger negative sandwiches.

Now, we state the recursive version of this algorithm. We remark that our betweenness reduction is an edge based guarantee (i.e., the nodes are weighted by their degree), which is slightly different from the guarantees from previous work. This was not needed in \cite{huang2026faster}, since they assumed (without loss of generality) that the graph has uniform degree. This becomes more complicated in our algorithm, since we recursively solve SSSP and need the uniform degree guarantee on the recursive instances as well. As a result, we write the degree-weighted version. 

\begin{definition}
    For a parameter $h\in\mathbb{N}$ and $s,t\in V$, the $h$-hop degree-betweenness of $s,t$ is the sum of the out-degrees $\deg^+(v)$ over all vertices $v$ such that $d^h(s,v)+d^h(v,t)<0$.
\end{definition}


\begin{lemma}\label{lem:new_bw_to_sssp}
    There is an algorithm that given a graph $G$ and parameters $h_0\ge1$, $h\ge h_0$, and $0<\alpha\le1$, either declares that $G$ contains a negative cycle, or find a valid reweighting such that
    \begin{itemize}
        \item for all $s,t\in V$, the $2h_0$-hop degree-betweenness of $(s,t)$ is at most $m/(h/h_0)^{1/\alpha}$, and
        \item for all $s,t\in V$ and $\eta\in[h_0,h]$, the $(\eta+h_0)$-hop degree-betweenness of $(s,t)$ is at most $m/(h/\eta)$.
    \end{itemize}
    The algorithm uses one oracle call to a SSSP instance with $O(mh_0)$ edges and $O((h/h_0)^{1/\alpha}\log n)$ negative edges and takes $\tilde O(mh)$ additional time.
\end{lemma}

\begin{proof}
    We now define the graph $H$ which we apply the SSSP oracle to. Let $G_{0}$, $G^{\tforw{}}_i$ for $i\in[h+h_0]$, and $G^{\tback{}}_i$ for $i\in[h+h_0]$ be copies of $G^+$, which is $G$ with all negative edges removed. For each vertex $v$, each copy of $v$ is referenced using the vertex name along with the subscripts/superscripts of the graph copy (e.g., $v_0$ or $v_5^{\tforw{}}$). 
    
    Let $M\in\mathbb{R}_{\ge0}$ be larger than the absolute value of all negative weights. For each negative edge $(u,v)$ in $G$ with weight $w(u,v)$, we add the following set of edges, each with weight $w(u,v)+M\ge 0$:
    \begin{itemize}
        \item an edge from $u_0$ to $v_1^{\tforw{}}$,
        \item an edge from $u_i^{\tforw{}}$ to $v_{i+1}^{\tforw{}}$ for each $i\in[h+h_0-1]$,
        \item an edge from $u_{i+1}^{\tback{}}$ to $v_{i}^{\tback{}}$ for each $i\in[h+h_0-1]$, and
        \item an edge from $u_1^{\tback{}}$ to $v_0$.
    \end{itemize} 

    Next, we define the negative edges in $H$. 
    Let $L=\lceil\log(h/h_0)\rceil$; we sample $L+1$ sets of edges $S_0,S_1,\ldots,S_{L}$, where edges $e\in E$ are sampled independently in each set. In the first set $S_0$, each edge is sampled with probability $p_0=\Theta((h/h_0)^{1/\alpha}\log(n)/m)$ and in the remaining sets $S_i$ for $i\ge 1$, each edge is sampled with probability $p_i=\Theta(h/(2^ih_0)\log(n)/m)$. For each $i\in[0,\ldots,L]$, let $T_i=\{x:(x,y)\in S_i\}$ denote the set of vertices which have an out-edge in $S_i$. By Chernoff bounds, we have $|S_0|=\Theta((h/h_0)^{1/\alpha}\log{n})$ and for $i\ge 1$, we have $|S_i|=\Theta(h/(2^ih_0)\log{n})$ with high probability. Thus, we also have $|T_0|=O((h/h_0)^{1/\alpha}\log{n})$ and for $i\ge 1$, we have $|T_i|=O(h/(2^ih_0)\log{n})$. Now, for each set $T_i$ and each vertex $v\in T_i$, we add an edge from $v^{\tforw{}}_{{2^ih_0+h_0}}$ to $v^{\tback{}}_{{2^ih_0+h_0}}$ with weight $-2M(2^ih_0+h_0)$. Finally, let $N_i$ denote the set of negative edges added for set $T_i$.
    
    

    For now, assume we have found some reweighting $\phi$ which neutralizes all the negative edges. Following the analysis in~\Cref{thm:faster_bw_reduction}, any negative cycle on the new graph indicates a negative cycle on $G$ and for any such reweighting $\phi$, we have the following for all pairs $s,t\in V$ and $i\in\{0,\ldots,L\}$:
    \begin{align}
        \phi(t_0)-\phi(s_0)\leq \min_{x\in T_i}\left(d^{2^ih_0+h_0}(s,x)+d^{2^ih_0+h_0}(x,t)\right)\label{eq:phi}
    \end{align}
    To verify $(\eta+h_0)$-hop degree-betweenness after the reweighting, fix $s,t\in V$ and a hop parameter $\eta\in[h_0,h]$. First consider $\eta=h_0$, and rank all the edges $(x,y)\in E$ in increasing order of $d^{2h_0}(s,x)+d^{2h_0}(x,t)$, with ties broken arbitrarily. With high probability, $S_0$ contains some edge $(x,y)$ with rank at most $m/(h/h_0)^{1/\alpha}$ between $s$ and $t$. For such $(x',y')$ later in the ordering, we will have the following after the reweighting, by \Cref{eq:phi}: 
    \begin{align*}
        d^{2h_0}_{\phi}(s,x')+d_{\phi}^{2h_0}(x',t)&=d^{2h_0}(s,x')+d^{2h_0}(x',t)+\phi(s_0)-\phi(t_0)\\
        &\ge d^{2h_0}(s,x')+d^{2h_0}(x',t)-\min_{x\in T_i}\left(d^{2^ih_0+h_0}(s,x)+d^{2^ih_0+h_0}(x,t)\right)\\
        &\ge d^{2h_0}(s,x')+d^{2h_0}(x',t)-d^{2^ih_0+h_0}(s,x)+d^{2^ih_0+h_0}(x,t)\\
        &\ge 0.
    \end{align*} 
    Thus, vertices $x'$ for which $d_{\phi}^{2h_0}(s,x')+d^{2h_0}_{\phi}(x',t)$ can be negative have all their out-edges within first $m/(h/h_0)^{1/\alpha}$ edges in the ordering. Thus, the $2h_0$-hop degree-betweenness is at most $m/(h/h_0)^{1/\alpha}$. For general $\eta\in[h_0,h]$, let $i=\lceil \log_2(h_0/\eta)\rceil$ and rank all the edges $(x,y)\in E$ in increasing order of $d^{2^ih_0+h_0}(s,x)+d^{2^ih_0+h_0}(x,t)$. With high probability, $S_i$ contains some edge $(x,y)$ with rank at most $m/(h/(2^ih_0))$ between $s$ and $t$. Via a similar argument as above, this implies that the $(2^ih_0+h_0)$-hop degree-betweenness is at most $m/(h/(2^ih_0))$.
    

    Now, we show how to compute a reweighting $\phi$ to neutralize the negative edges. To compute the shortest path on the constructed graph, we use the following procedure. To simplify notation, we use $G^{\tforw{}}_{{[l,r]}}$ to denote the union of $G^{\tforw{}}_{{i}}$ for $i\in \{l,\ldots,r\}$ together with the edges between them, and analogously, $G^{\tback{}}_{{[l,r]}}$ to denote the union of $G^{\tback{}}_{{i}}$ for $i\in\{l,\ldots,r\}$ with the edges between them.

    \paragraph{The Algorithm.} For $i=0,1,\ldots,L$ in order, do the following:
    \begin{enumerate}
        \item Let $H_i$ denote the induced subgraph $H$ that contains $G_{0}$, $G^{\tforw{}}_{{[1,2^ih_0+h_0]}}$, and $G^{\tback{}}_{{[1,2^ih_0+h_0]}}$. Using an SSSP oracle (to be chosen later), compute potentials $\phi_i(v)=d_{H_i}(V,v)$ for each $v\in V$. If this step reports a negative cycle, declare that $G$ has a negative cycle.
        \item Let $\phi_i^{\max}=\max_{v\in V(H_i)}\phi_i(v)$ and $\phi_i^{\min}=\min_{v\in V(H_i)}\phi_i(v)$. Extend the potential $\phi_i$ from $H_i$ to the whole graph $H$ in the following way: for vertices $v$ in $G^{\tforw{}}_{{[2^ih_0+h_0+1,h+h_0]}}$, set their potential as $\phi_i(v)=\phi_i^{\min}$ and for vertices $v$ in $G^{\tback{}}_{{[2^ih_0+h_0+1,h+h_0]}}$, set their potential as $\phi_i(v)=\phi_i^{\max}$. Use this potential to reweight the whole graph.
    \end{enumerate}

    \paragraph{Analysis.} We claim that after the reweighting in the $i^{th}$ round, the only edges in $H_{\phi_i}$ which may be negative are $N_{i+1},N_{i+2},\ldots,N_{L}$. This can be proved by induction, and the non-negativeness of other edges in each step can be verified easily:

    \begin{itemize}
        \item The edges in the induced graph are non-negative because they use the potential computed on the induced graph;
        \item The edges from $G^{\tforw{}}_{{2^ih_0+h_0}}$ to $G^{\tforw{}}_{{2^ih_0+h_0+1}}$ are non-negative because these edges are non-negative after the previous round, and the potential in $G^{\tforw{}}_{{2^ih_0+h_0+1}}$ equals $\phi_{\min}$, which is not larger than potentials in $G^{\tforw{}}_{{2^ih_0+h_0}}$;
        \item The edges in $G^{\tforw{}}_{{[2^ih_0+h_0+1,h+h_0]}}$ are non-negative because vertices in this part have the same potential, and these edges are non-negative after the previous round;
        \item The edges from $G^{\tback{}}_{{2^ih_0+h_0+1}}$ to $G^{\tback{}}_{{2^ih_0+h_0}}$, and the edges in $G^{\tback{}}_{{[2^ih_0+h_0+1,h+h_0]}}$ can be proved similar to the argument above.
    \end{itemize}

    We now choose the algorithm for the SSSP instances at each step. When $i=0$, the instance contains $|N_0|=|S_0|=O((h/h_0)^{1/\alpha}\log n)$ negative edges and $O(mh_0)$ edges, and we use the SSSP oracle to solve this instance. When $i>1$, the instance contains $|N_i|=|S_i|=O(h/(2^ih_0)\log n)$ negative edges and $O(m2^ih_0)$ edges, and we use a trivial algorithm for it, which runs in time $\tilde O((m2^ih_0)\cdot (h/(2^ih_0)\log n))=\tilde O(mh)$. Therefore, the total running time (excluding the oracle call) is $\tilde O(mh)$.
\end{proof}

\subsection{Algorithm}

We now describe the recursive algorithm, with the idea of converting betweenness reductions and hop reducer constructions on lower levels to solving shortest path problems, and skipping lower levels of hop reducer constructions.

\paragraph{The algorithm} Let $m$ denote the number of edges, and $k$ denote the number of negative edges. If $k=m^{o(1)}$, we use the trivial $\tilde O(mk)$ algorithm directly. Otherwise, we repeat the following steps to neutralize negative edges until only $m^{o(1)}$ negative edges remain, and then solve them directly. Let $\alpha\approx0.694\ldots$ be the root of $2\alpha^3-3\alpha^2+4\alpha-2=0$, and let $h=k^{1-\alpha}$ and $h_0=k^{3-4\alpha}$ be parameters that will be used in the following steps. In each round, we do the following.

\begin{enumerate}
    \item Apply \Cref{lem:new_bw_to_sssp} with parameters $h,h_0$ and recursively use this algorithm as the oracle for the SSSP problem. This ensures that the $2h_0$-hop betweenness is $O(n/(h/h_0)^{1/\alpha})$, and for any $\eta\in[h_0,h]$, the $(\eta+h_0)$-hop betweenness is $O(n/(h/\eta))$.
    \item Apply \Cref{lem:apx-big-sandwich} with our $h_0$. We either neutralize $\sqrt{kh_0}$ edges directly, or find a $h_0$-hop negative sandwich $(s,U,t)$ of size $\sqrt{kh_0}$. In the first case, we can directly finish this round. The remaining steps aim to neutralize the sandwich in the second case. 
    \item Apply \Cref{lem:hjq26_bw_to_neg_reach} to convert betweenness guarantees to bounds on the negative reach of $U$. After this, we have that the $h_0$-hop negative reach of $U$ has at most $O(n/(h/h_0)^{1/\alpha})$ vertices, and for any $\eta\in[h_0,h]$, the $\eta$-hop negative reach has $O(n/(h/\eta))$ vertices.
    \item Let $G=G_U$ after the reweightings in the previous steps. For each $i\in[L]$, let $V_i$ denote the $h_02^i$-hop negative reach of $U$ and let $G_i$ denote the induced subgraph $G[V_i]$. By  \Cref{lem:apx-big-sandwich,lem:hjq26_bw_to_neg_reach}, $G_0$ has $m_0=O(m/(h/h_0)^{1/\alpha})$ edges and $G_i$ has $m_i=O(2^ih_0m/h)$ edges.
    \item Apply \Cref{cl:easy-hop-reducer} to the graph $G_0$, where we again recursively use this algorithm as the oracle for the SSSP instance. This provides an $|U|$-hop reducer for this $G_0$, which has size  $O(m/(h/h_0)^{1/\alpha})$. 
    \item Apply the remaining steps of the bootstrapping construction \Cref{lem:compute-distance-estimates,lem:construct-hop-reducer} iteratively as before, starting with the $|U|$-hop reducer for $G_0$ to get a $h$-hop reducer for $G_U$ of $\tilde{O}(m+|U|^2)$ edges, and neutralize the sandwich using the hop reducer.
\end{enumerate}


\paragraph{Analysis.} We first analyze the running time, ignoring the recursive calls. Each round neutralizes $\sqrt{h_0k}=k^{2-2\alpha}$ negative edges, so the algorithm will finish in $\tilde O(k^{2\alpha-1})$ rounds. Within each round, step 1 takes $O(mh\log{m})$ time, step 2 takes $O(mh_0\log{m})$ time, and step 3 takes $O(mh\log{m})$ time. Finally, step 6 takes $\tilde O(m\sqrt{h_0k}/h+(\sqrt{h_0k})^3/h_0)=\tilde O(mk^{1-\alpha}+k^{3-2\alpha})$ time. Therefore, the total running time without recursion is $\tilde O(k^2+mk^\alpha)$.

Now we consider the two recursive calls made by the algorithm. The first call in \Cref{lem:new_bw_to_sssp} is on an instance with $O(mh_0)=O(mk^{3-4\alpha})$ edges and $O((h/h_0)^{1/\alpha}\log n)=O(k^{(3\alpha-2)/\alpha}\log n)$ negative edges. The second call for hop reducers is on an instance with $O(m/(h/h_0)^{1/\alpha})=O(m/k^{{(3\alpha-2)/\alpha}})$ edges and $|U|=\sqrt{h_0k}=O(k^{2-2\alpha})$ negative edges. Letting $C\ge1$ be a large enough constant and $T(m,k)$ denote the running time of our algorithm on a graph with $m$ edges and $k$ negative edges, we have the following recursion. 
\[T(m,k)\leq Cm^{o(1)}(mk^\alpha+k^2)+Ck^{2\alpha-1}\log m\left(T(Cmk^{3-4\alpha},Ck^{(3\alpha-2)/\alpha}\log m)+T(Cm/k^{(3\alpha-2)/\alpha},Ck^{2-2\alpha})\right).\]

To analyze this recursion, we consider each node in the recursion tree, and use $m_p$ and $k_p$ to denote the number of edges and negative edges, respectively, at a node $p$. In recursions, we use them to refer to the parameters in the current recursion, and we use $m,k$ to refer to the parameters in the main graph. Moreover, we define the \emph{coefficient} $c_p$ on a node to be the product of all leading factors $Ck^{2\alpha-1}$ encountered from the root of the recursion tree to node $p$, which upper bounds the number of times this node is actually visited. Coefficients give a simple way to count the total running time: Each time a node $p$ is visited in the recursion, the non-recursive part in this node takes $Cm_p^{o(1)}(m_pk_p^\alpha+k_p^2)$ time, so summing the $c_p$ times this over all nodes in the recursion tree gives the total running time, as follows.

\begin{align}
T(m,k)&\leq \sum_p C_p\cdot Cm_p^{o(1)}(m_pk_p^\alpha+k_p^2).\label{eq:recursion_coef}
\end{align}

First, we analyze the possible range of these parameters in different depths of the recursion tree. Here, the root of the recursion tree has depth $0$.

\begin{fact}\label{lem:recursion_bound_para}
    For any $d$ and any node $p$ on the recursion tree with depth $d\leq \log\log m$, we have
    \begin{align*}
        m_p&\leq mk^{1-(2/3)^d}(3C\log m)^{3d},\\
        k_p&\leq k^{(2/3)^d}(3C\log m)^3.
    \end{align*}
    Moreover,
    \[m_p\geq m/\left(k^{(1-(2/3)^d)/2}(3C\log m)^{d/2}\right).\]
\end{fact}

\begin{proof}
    We proceed by induction on $d$. The case when $d=0$ trivially holds. Now suppose that $p$ is a node at depth $d$, and $q$ is a direct child of $p$. Then we have
    \begin{align*}
        m_q&\leq \max\{Cm_pk_p^{3-4\alpha},Cm_p/k_p^{(3\alpha-2)/\alpha}\}=Cm_pk_p^{3-4\alpha}\leq Cm_pk_p^{1/3},\\
        k_q&\leq \max\{Ck_p^{(3\alpha-2)/\alpha}\log m_p,Ck_p^{2-2\alpha}\}\leq Ck_p^{2-2\alpha}\log m_p\leq Ck_p^{2/3}\log m_p.
    \end{align*}
    Here we use $\alpha\approx 0.694\ldots$. Now, by the induction hypothesis, we have
    \begin{align*}
        m_q&\leq C\left(mk^{1-(2/3)^d}(3C\log m)^{3d}\right)\left(k^{(2/3)^d}(3C\log m)^3\right)^{1/3}\\
        &=Cmk^{1-(2/3)^d+(2/3)^d/3}(3C\log m)^{3d}(3C\log m)\\
        &\leq mk^{1-(2/3)^{d+1}}(3C\log m)^{3(d+1)}.
    \end{align*}
    When $d\leq \log\log m$, we have $m_p\leq mk(3C\log m)^{3d}\leq m^3$. Using this bound on $m_p$, we have
    \begin{align*}
        k_q&\leq C\left(k^{(2/3)^d}(3C\log m)^3\right)^{2/3}(3\log m)\\
        &=k^{(2/3)^{d+1}}(3C\log m)^3.
    \end{align*}
    For the lower bound, we can similarly have $m_q\geq m_p/k_p^{1/6}$, so
    \begin{align*}
        m_q&\geq \left(m/\left(k^{(1-(2/3)^d)/2}(3C\log m)^{d/2}\right)\right)/\left(k^{(2/3)^d}(3C\log m)^3\right)^{1/6}\\
        &=m/\left(k^{(1-(2/3)^d)/2+(2/3)^d/6}(3C\log m)^{d/2+1/2}\right)\\
        &=m/\left(k^{(1-(2/3)^{d+1})/2}(3C\log m)^{(d+1)/2}\right).\qedhere
    \end{align*}
\end{proof}

The above upper bounds and lower bounds on parameters also imply a bound on the recursion tree, as follows.

\begin{corollary}\label{cl:recursion_depth}
    Let $p$ be a node on the recursion tree with depth $d=\log \log m$, then 
    \begin{align*}
        m_p&\geq m/\left(k^{(1-(2/3)^d)/2}(3C\log m)^{d/2}\right)\geq m^{1/3},\\
        k_p&\leq k^{(2/3)^d}(3C\log m)^3\leq m^{o(1)}.
    \end{align*}
    Therefore, the recursion always terminates before $\log\log m$ depth.
\end{corollary}

Therefore, we only need to analyze the first $\log\log m$ levels in the recursion tree. In this range, we have $m_p\leq m^3$ for all nodes $p$ in the recursion tree, so we can safely replace all $\log m_p$ with $3\log m$ when analyzing the upper bound of running time. 

We now bound the total running time in \eqref{eq:recursion_coef}. There are two terms $\sum_p c_pm_pk_p^\alpha$ and $\sum_p c_pk_p^2$ in the running time, and we analyze them separately. For the first term, we have

\begin{fact}\label{lem:recursion_bound_mka}
    Let $p,q$ be two nodes on the recursion tree such that $q$ is a direct child of $p$. Then,
    \[c_qm_qk_q^\alpha\leq C^3(3\log m)^2\cdot c_pm_pk_p^\alpha.\]
\end{fact}

\begin{proof}
    Suppose that $q$ corresponds to the first term $T(Cmk^{3-4\alpha},Ck^{(3\alpha-2)/\alpha}\log m)$ in the recursion, then
    \begin{align*}
        c_q&\leq Ck_p^{2\alpha-1}(3\log m)c_p,\\
        m_q&\leq Cm_pk_p^{3-4\alpha},\\
        k_q&\leq Ck_p^{(3\alpha-2)/\alpha}(3\log m).
    \end{align*}

    Therefore,
    \begin{align*}
        c_qm_qk_q^\alpha&\leq Ck_p^{2\alpha-1}(3\log m)c_p\cdot Cm_pk_p^{3-4\alpha}\cdot \left(Ck_p^{(3\alpha-2)/\alpha}(3\log m)\right)^\alpha\\
        &=C^{2+\alpha}(3\log m)^{1+\alpha} c_pm_pk_p^\alpha\\
        &\leq C^3(3\log m)^2\cdot c_pm_pk_p^\alpha
    \end{align*}
    because $\alpha<1$.

    Now, suppose that $q$ corresponds to the second term, then
    \begin{align*}
        c_q&\leq Ck_p^{2\alpha-1}(3\log m)c_p,\\
        m_q&\leq Cm_p/k_p^{(3\alpha-2)/\alpha},\\
        k_q&\leq Ck_p^{2-2\alpha}.
    \end{align*}
    Therefore,
    \begin{align*}
        c_qm_qk_q^\alpha&\leq Ck_p^{2\alpha-1}(3\log m)c_p\cdot Cm_p/k_p^{(3\alpha-2)/\alpha}\cdot (Ck_p^{2-2\alpha})^\alpha\\
        &=C^{2+\alpha}c_pm_pk_p^{-2\alpha^2+4\alpha-4+2/\alpha}\\
        &=C^{2+\alpha}c_pm_pk_p^{\alpha}\\
        &\leq C^3(3\log m)^2\cdot c_pm_pk_p^\alpha,
    \end{align*}
    where the third line uses the fact that $\alpha$ is the root of equation $2\alpha^3-3\alpha^2+4\alpha-2=0$.
\end{proof}

The above fact shows that $c_qm_qk_q^\alpha$ can be bounded by a function of the depth $d$ of the node inductively. Because we have $d\leq \log\log m$ in \Cref{cl:recursion_depth}, we can later show that $c_qm_qk_q^\alpha$ never exceeds $m^{1+o(1)}k^\alpha$ in the recursion tree, which implies that the running time incurred by the $mk^\alpha$ term is $O(m^{1+o(1)}k^\alpha)$ because there are at most $O(2^d)=O(\log n)$ nodes in the recursion tree.

\begin{corollary}\label{cl:recursion_bound_mka}
    Let $p$ be a node on the recursion tree of depth $d$ (recall that the root of the tree has depth $0$). Then,
    \[c_pm_pk_p^\alpha\leq C^{3d}(3\log m)^{2d}mk^\alpha.\]
\end{corollary}

We proceed to analyze the $k^2$ term in the recursion, similar to the analysis above.

\begin{fact}\label{lem:recursion_bound_k2}
    Let $p,q$ be two nodes on the recursion tree such that $q$ is a direct child of $p$. Then,
    \[c_qk_q^2\leq C^3(3\log m)^3 c_pk_p^2.\]
\end{fact}

\begin{proof}
    Suppose that $q$ corresponds to the first term in the recursion, then
    \begin{align*}
        c_q&\leq Ck_p^{2\alpha-1}(3\log m)c_p,\\
        k_q&\leq Ck_p^{(3\alpha-2)/\alpha}(3\log m).
    \end{align*}
    Therefore,
    \begin{align*}
        c_qk_q^2&\leq Ck_p^{2\alpha-1}(3\log m)c_p\cdot \left(Ck_p^{(3\alpha-2)/\alpha}(3\log m)\right)^2\\
        &=C^3(3\log m)^3c_pk_p^{2\alpha-1+2(3\alpha-2)/\alpha}\\
        &\leq C^3(3\log m)^3c_pk_p^2.
    \end{align*}

    Now, suppose that $q$ corresponds to the second term, then
    \begin{align*}
        c_q&\leq Ck_p^{2\alpha-1}(3\log m)c_p,\\
        k_q&\leq Ck_p^{2-2\alpha}.
    \end{align*}
    Therefore,
    \begin{align*}
        c_qk_q^2&\leq Ck_p^{2\alpha-1}(3\log m)c_p\cdot (Ck_p^{2-2\alpha})^2\\
        &=C^3(3\log m)c_pk_p^{3-2\alpha}\\
        &\leq C^3(3\log m)^3\cdot c_pk_p^2.\qedhere
    \end{align*}
\end{proof}

\begin{corollary}\label{cl:recursion_bound_k2}
    Let $p$ be a node on the recursion tree of depth $d$. Then,
    \[c_pk_p^2\leq C^{3d}(3\log m)^{3d}k^2.\]
\end{corollary}

Now, we can bound the total running time using all the analyses above.

\begin{lemma}
    In the above recursion, we have $T(m,k)=O(m^{o(1)}(mk^\alpha+k^2))$.
\end{lemma}

\begin{proof}
    By \Cref{cl:recursion_depth}, we know that any vertex $p$ is of depth at most $\log\log m$. Using this with the above facts, we have
    \begin{align*}
        T(m,k)&\leq C\sum_p c_pm_p^{o(1)}(m_pk_p^\alpha+k_p^2)\tag{By \Cref{eq:recursion_coef}}\\
        &\leq Cm^{o(1)}\sum_p c_p(m_pk_p^\alpha+k_p^2)\tag{\Cref{lem:recursion_bound_para}, $m_p\leq m^3$}\\
        &\leq Cm^{o(1)}\sum_p C^{3\log\log m}(3\log m)^{3\log\log m}(mk^\alpha+k^2) \tag{\Cref{cl:recursion_bound_mka}, \Cref{cl:recursion_bound_k2}}\\
        &\leq m^{o(1)}\sum_p(mk^\alpha+k^2)\\
        &\leq m^{o(1)}2^{\log\log m+1}(mk^\alpha+k^2)\tag{The number of nodes is at most $2^{d+1}$}\\
        &=m^{o(1)}(mk^\alpha+k^2).\qedhere
    \end{align*}

\end{proof}

With this set of parameters, we get an algorithm with running time $m^{o(1)}(mk^\alpha+k^2)$, which is $m^{1+o(1)}k^\alpha$ when the graph is not too sparse ($m\geq k^{2-\alpha}$ in particular). When the graph is sparse, we can slightly change our algorithm to achieve a better complexity for sparse graphs. In this case, the $k^2$ term would become a bottleneck in this algorithm, which comes from the $|U|^3$ term in constructing the hop reducer. To reduce this term, in addition to choosing a different set of parameters, one change is that after finding a negative sandwich of size $\sqrt {h_0k}$, we only neutralize part of these negative edges to reduce the $|U|^3$ term.

\paragraph{Algorithm for sparse graphs where $m\leq k^{2-\alpha}$.} We make the following changes to the algorithm when the input graph has $m\leq k^{2-\alpha}$. First, we choose $h=m^{(1-\alpha)/(2-\alpha)}$ and $h_0=m^{(3-4\alpha)/(2-\alpha)}$. Moreover, when the algorithm finds a negative sandwich with $\sqrt{h_0k}\geq m^{(2-2\alpha)/(2-\alpha)}$ negative edges, we only pick a subset of negative edges with size $|U|=m^{(2-2\alpha)/(2-\alpha)}$ for further neutralization.

\paragraph{Analysis for sparse graphs.} We first analyze the root case of the recursion. In each round, the algorithm neutralizes $m^{(2-2\alpha)/(2-\alpha)}$ edges with the following work.
\begin{itemize}
    \item For the non-recursive part, \Cref{lem:new_bw_to_sssp} takes $\tilde O(mh)$ time (without considering the recursion call), and \Cref{lem:apx-big-sandwich,lem:hjq26_bw_to_neg_reach} take $\tilde O(mh_0)$ time. The remaining steps of hop reducer construction take $\tilde O(m|U|/h+|U|^3/h_0)$ time. Again, $h_0\leq h$, and the total running time without recursion is
    \[\tilde O\left(m^{1+(1-\alpha)/(2-\alpha)}+(m^{(2-2\alpha)/(2-\alpha)})^3/m^{(3-4\alpha)/(2-\alpha)}\right)=\tilde O(m^{(3-2\alpha)/(2-\alpha)}).\]
    \item The recursion call in betweenness reduction (\Cref{lem:new_bw_to_sssp}) uses a graph with $O(mh_0)=O(m^{1+(3-4\alpha)/(2-\alpha)})$ edges and $O((h/h_0)^{1/\alpha}\log n)=O(m^{(3\alpha-2)/(\alpha(2-\alpha))}\log n)$ negative edges. For this instance, the value of $k^{2-\alpha}$ is at most
    \[\left(m^{(3\alpha-2)/(\alpha(2-\alpha))}\log n\right)^{2-\alpha}=\tilde O(m^{(3\alpha-2)/\alpha})<O(m^{1+(3-4\alpha)/(2-\alpha)})\]
    because $(3\alpha-2)/\alpha\approx 0.11\ldots$ and $1+(3-4\alpha)/(2-\alpha)\approx 1.17\ldots$. So this instance is in the dense case $m>k^{2-\alpha}$. Solving this by the original algorithm takes
    \[\tilde O\left(m^{o(1)}m^{1+(3-4\alpha)/(2-\alpha)}\cdot m^{\left((3\alpha-2)/(\alpha(2-\alpha))\right)\cdot\alpha}\right)=\tilde O(m^{(3-2\alpha)/(2-\alpha)+o(1)})\]
    time.
    \item The recursion call in hop reducers uses an instance with $O(m/(h/h_0)^{1/\alpha})=O(m^{1-(3\alpha-2)/(\alpha(2-\alpha))})$ edges and $|U|=m^{(2-2\alpha)/(2-\alpha)}$ negative edges. For this instance,
    \[|U|^{2-\alpha}=m^{2-2\alpha}<m^{0.62}<m^{0.9}<m^{1-(3\alpha-2)/(\alpha(2-\alpha))},\]
    so it also fits in the dense case, and the running time is
    \begin{align*}
        &~\tilde O\left(m^{o(1)}m^{1-(3\alpha-2)/(\alpha(2-\alpha))}\cdot m^{\alpha(2-2\alpha)/(2-\alpha)}\right)\\
        =&~\tilde O\left(m^{o(1)}m^{\frac{(2-\alpha-\alpha^2)+\alpha^2(2-2\alpha)}{\alpha(2-\alpha)}}\right)\\
        =&~\tilde O\left(m^{o(1)}m^{\frac{2-\alpha+\alpha^2-2\alpha^3}{\alpha(2-\alpha)}}\right)\\
        =&~\tilde O\left(m^{o(1)}m^{\frac{3\alpha-2\alpha^2}{\alpha(2-\alpha)}}\right)\\
        =&~\tilde O(m^{(3-2\alpha)/(2-\alpha)+o(1)}),
    \end{align*}
    where the fourth line uses the fact that $2\alpha^3-3\alpha^2+4\alpha-2=0$.
\end{itemize}

Therefore, the algorithm always neutralizes $m^{(2-2\alpha)/(2-\alpha)}$ negative edges in $O(m^{(3-2\alpha)/(2-\alpha)+o(1)})$ time, so the total running time is $O(m^{1/(2-\alpha)+o(1)}k)$. 

Combining these two cases, we get

\begin{theorem}
    Let $\alpha$ be the only real root of the equation $2\alpha^3-3\alpha^2+4\alpha-2=0$. There is an algorithm that solves the real-weight SSSP problem with $m$ edges and $k$ negative edges in $O(m^{o(1)}(mk^{\alpha}+m^{1/(2-\alpha)}k))$ time with high probability.
\end{theorem}

\Cref{thm:main-sparse} follows from $\alpha\approx0.694\ldots$ and the trivial bound $k\le n$.

\end{document}